\documentclass{IEEEtran}
\ifCLASSINFOpdf
\else
\fi
\usepackage{url}
\usepackage[english]{babel}
\usepackage[autostyle, english=american]{csquotes}
\MakeOuterQuote{"}

\usepackage{amsbsy}
\usepackage{amsmath}
\usepackage{array}
\usepackage{fixmath}
\usepackage{amssymb}
\usepackage{epsfig}
\usepackage{epstopdf}
\usepackage{algorithmic}
\usepackage{algorithm}
\usepackage{amsthm}

\newtheorem{theorem}{Theorem}

\newtheorem{lemma}[theorem]{Lemma}

\newtheorem{definition}{Definition}
\newtheorem{remark}{Remark}

\usepackage{tikz}
\usetikzlibrary{arrows,decorations,backgrounds,shadows,plotmarks,calc, positioning}
\usetikzlibrary{arrows.meta}

\usepackage{pgfplots}
\pgfplotsset{compat=newest}
\pgfplotsset{plot coordinates/math parser=false}
\pgfplotsset{every axis/.append style={font=\footnotesize}}
\pgfplotsset{
    ylabel right/.style={
        after end axis/.append code={
            \node [rotate=90, anchor=north] at (rel axis cs:1,0.5) {#1};
        }   
    }
}
\newlength\figureheight
\newlength\figurewidth
\newlength\subgraphheight
\newlength\subgraphwidth
\makeatletter
\newcommand{\nodedistance}{\tikz@node@distance}
\makeatother

\makeatletter

\pgfkeys{
	/tikz/blockm/.initial=1,
	/tikz/blockn/.initial=1
}

\pgfdeclareshape{dspantenna}{
	
	\inheritsavedanchors[from=rectangle]
	
	\inheritanchorborder[from=rectangle]
	
	\inheritanchor[from=rectangle]{center}
	\inheritanchor[from=rectangle]{south}
	\inheritanchor[from=rectangle]{west}
	\inheritanchor[from=rectangle]{north}
	\inheritanchor[from=rectangle]{east}
	\inheritanchor[from=rectangle]{south east}
	\inheritanchor[from=rectangle]{south west}
	\inheritanchor[from=rectangle]{north east}
	\inheritanchor[from=rectangle]{north west}
	\saveddimen{\halfheight}{
		\pgfmathsetlength{\pgf@xa}{\pgfkeysvalueof{/pgf/minimum height}/2}
		\pgfmathsetlength{\pgf@xb}{\pgfkeysvalueof{/pgf/outer ysep}}
		\advance\pgf@xa by \pgf@xb
		\pgf@x=\pgf@xa
	}
	\saveddimen{\halfwidth}{
		\pgfmathsetlength{\pgf@xa}{\pgfkeysvalueof{/pgf/minimum width}/2}
		\pgfmathsetlength{\pgf@xb}{\pgfkeysvalueof{/pgf/outer xsep}}
		\advance\pgf@xa by \pgf@xb
		\pgf@x=\pgf@xa
	}
	
	\backgroundpath{
		
		\northeast \pgf@xa=\pgf@x \pgf@ya=\pgf@y
		\pgf@xb=\pgf@xa \pgf@yb=\pgf@ya
		
		\pgfpathmoveto{\pgfpoint{\pgf@xb}{\pgf@yb}}
		\advance \pgf@xb by -\halfwidth
		\advance \pgf@yb by -\halfwidth
		\pgfpathlineto{\pgfpoint{\pgf@xb}{\pgf@yb}}
		\advance \pgf@xb by -\halfwidth
		\advance \pgf@yb by \halfwidth
		\pgfpathlineto{\pgfpoint{\pgf@xb}{\pgf@yb}}
		\advance \pgf@xb by \halfwidth
		\advance \pgf@yb by -\halfwidth
		\pgfpathmoveto{\pgfpoint{\pgf@xb}{\pgf@yb}}
		\advance \pgf@yb by -\halfwidth
		\pgfpathlineto{\pgfpoint{\pgf@xb}{\pgf@yb}}
		\advance \pgf@xb by \halfwidth
		\pgfpathlineto{\pgfpoint{\pgf@xb}{\pgf@yb}}
		
		\pgfusepath{draw}
	}
}

\pgfdeclareshape{dspvdots}{
	
	\inheritsavedanchors[from=rectangle]
	
	\inheritanchorborder[from=rectangle]
	
	\inheritanchor[from=rectangle]{center}
	\inheritanchor[from=rectangle]{south}
	\inheritanchor[from=rectangle]{west}
	\inheritanchor[from=rectangle]{north}
	\inheritanchor[from=rectangle]{east}
	\inheritanchor[from=rectangle]{south east}
	\inheritanchor[from=rectangle]{south west}
	\inheritanchor[from=rectangle]{north east}
	\inheritanchor[from=rectangle]{north west}
	\saveddimen{\halfheight}{
		\pgfmathsetlength{\pgf@xa}{\pgfkeysvalueof{/pgf/minimum height}/2}
		\pgfmathsetlength{\pgf@xb}{\pgfkeysvalueof{/pgf/outer ysep}}
		\advance\pgf@xa by \pgf@xb
		\pgf@x=\pgf@xa
	}
	\saveddimen{\halfwidth}{
		\pgfmathsetlength{\pgf@xa}{\pgfkeysvalueof{/pgf/minimum width}/2}
		\pgfmathsetlength{\pgf@xb}{\pgfkeysvalueof{/pgf/outer xsep}}
		\advance\pgf@xa by \pgf@xb
		\pgf@x=\pgf@xa
	}
	
	\backgroundpath{
		
		\northeast \pgfutil@tempdima=\pgf@x \pgfutil@tempdimb=\pgf@y
		
		\pgf@xa=\halfwidth
		\pgf@ya=\halfheight
		\advance\pgfutil@tempdima by -\pgf@xa
		\advance\pgfutil@tempdimb by -\pgf@ya 
		
		\pgfpathcircle{\pgfpoint{\pgfutil@tempdima}{\pgfutil@tempdimb}}{\halfwidth/3}
		\pgf@xa=\halfwidth
		\pgf@ya=\halfheight
		\advance\pgfutil@tempdimb by \pgf@xa
		\advance\pgfutil@tempdimb by -\pgf@ya
		\pgfpathcircle{\pgfpoint{\pgfutil@tempdima}{\pgfutil@tempdimb}}{\halfwidth/3}
		\pgf@xa=\halfwidth
		\pgf@ya=\halfheight
		\advance\pgfutil@tempdimb by 2\pgf@ya
		\advance\pgfutil@tempdimb by -2\pgf@xa
		\pgfpathcircle{\pgfpoint{\pgfutil@tempdima}{\pgfutil@tempdimb}}{\halfwidth/3} 
		
		\pgfusepath{fill}
	}
}

\pgfdeclareshape{dspadder}{
	\inheritsavedanchors[from=circle]
	\inheritanchorborder[from=circle]
	
	\inheritanchor[from=circle]{center}
	\inheritanchor[from=circle]{south}
	\inheritanchor[from=circle]{west}
	\inheritanchor[from=circle]{north}
	\inheritanchor[from=circle]{east}
	\inheritanchor[from=circle]{south east}
	\inheritanchor[from=circle]{south west}
	\inheritanchor[from=circle]{north east}
	\inheritanchor[from=circle]{north west}
	
	\inheritbackgroundpath[from=circle]
	
	\beforebackgroundpath{
		
		\radius \pgf@xb=\pgf@x
		\centerpoint \pgf@xa=\pgf@x \pgf@ya=\pgf@y \pgf@xc=\pgf@x \pgf@yc=\pgf@y
		
		\advance\pgf@yc by \pgf@xb
		\pgfpathmoveto{\pgfpoint{\pgf@xc}{\pgf@yc}}
		\advance\pgf@yc by -2\pgf@xb
		\pgfpathlineto{\pgfpoint{\pgf@xc}{\pgf@yc}}
		\advance\pgf@xc by -\pgf@xb
		\advance\pgf@yc by \pgf@xb
		\pgfpathmoveto{\pgfpoint{\pgf@xc}{\pgf@yc}}
		\advance\pgf@xc by 2\pgf@xb
		\pgfpathlineto{\pgfpoint{\pgf@xc}{\pgf@yc}}
		
		\pgfusepath{draw}
	}
}

\pgfdeclareshape{dsptriangler}{
	
	\inheritsavedanchors[from=rectangle]
	
	\inheritanchorborder[from=rectangle]
	
	\inheritanchor[from=rectangle]{center}
	\inheritanchor[from=rectangle]{south}
	\inheritanchor[from=rectangle]{west}
	\inheritanchor[from=rectangle]{north}
	\inheritanchor[from=rectangle]{east}
	\inheritanchor[from=rectangle]{south east}
	\inheritanchor[from=rectangle]{south west}
	\inheritanchor[from=rectangle]{north east}
	\inheritanchor[from=rectangle]{north west}
	\saveddimen{\halfheight}{
		\pgfmathsetlength{\pgf@xa}{\pgfkeysvalueof{/pgf/minimum height}/2}
		\pgfmathsetlength{\pgf@xb}{\pgfkeysvalueof{/pgf/outer ysep}}
		\advance\pgf@xa by \pgf@xb
		\pgf@x=\pgf@xa
	}
	\saveddimen{\halfwidth}{
		\pgfmathsetlength{\pgf@xa}{\pgfkeysvalueof{/pgf/minimum width}/2}
		\pgfmathsetlength{\pgf@xb}{\pgfkeysvalueof{/pgf/outer xsep}}
		\advance\pgf@xa by \pgf@xb
		\pgf@x=\pgf@xa
	}
	
	\backgroundpath{

		\pgfutil@tempdima=\halfheight
		\pgfutil@tempdimb=\halfwidth
		\southwest \pgf@xa=\pgf@x \pgf@ya=\pgf@y 
		\pgfpathmoveto{\pgfpoint{\pgf@xa}{\pgf@ya}}
		\advance \pgf@ya by 2\pgfutil@tempdima
		\pgfpathlineto{\pgfpoint{\pgf@xa}{\pgf@ya}}
		\advance \pgf@xa by 2\pgfutil@tempdimb
		\advance \pgf@ya by -\pgfutil@tempdima
		\pgfpathlineto{\pgfpoint{\pgf@xa}{\pgf@ya}}	
		\advance \pgf@xa by -2\pgfutil@tempdimb
		\advance \pgf@ya by -\pgfutil@tempdima
		\pgfpathlineto{\pgfpoint{\pgf@xa}{\pgf@ya}}       
		\pgfusepath{draw}
	}
}

\pgfdeclareshape{dspdiamond}{
	
	\inheritsavedanchors[from=rectangle]
	
	\inheritanchorborder[from=rectangle]
	
	\inheritanchor[from=rectangle]{center}
	\inheritanchor[from=rectangle]{south}
	\inheritanchor[from=rectangle]{west}
	\inheritanchor[from=rectangle]{north}
	\inheritanchor[from=rectangle]{east}
	\inheritanchor[from=rectangle]{south east}
	\inheritanchor[from=rectangle]{south west}
	\inheritanchor[from=rectangle]{north east}
	\inheritanchor[from=rectangle]{north west}
	\saveddimen{\halfheight}{
		\pgfmathsetlength{\pgf@xa}{\pgfkeysvalueof{/pgf/minimum height}/2}
		\pgfmathsetlength{\pgf@xb}{\pgfkeysvalueof{/pgf/outer ysep}}
		\advance\pgf@xa by \pgf@xb
		\pgf@x=\pgf@xa
	}
	\saveddimen{\halfwidth}{
		\pgfmathsetlength{\pgf@xa}{\pgfkeysvalueof{/pgf/minimum width}/2}
		\pgfmathsetlength{\pgf@xb}{\pgfkeysvalueof{/pgf/outer xsep}}
		\advance\pgf@xa by \pgf@xb
		\pgf@x=\pgf@xa
	}
	
	\backgroundpath{

		\pgfutil@tempdima=\halfheight
		\pgfutil@tempdimb=\halfwidth
		\southwest \pgf@xa=\pgf@x \pgf@ya=\pgf@y 
		\advance\pgf@xa by \pgfutil@tempdimb
		\pgfpathmoveto{\pgfpoint{\pgf@xa}{\pgf@ya}}
		\advance \pgf@xa by \pgfutil@tempdimb
		\advance \pgf@ya by \pgfutil@tempdima
		\pgfpathlineto{\pgfpoint{\pgf@xa}{\pgf@ya}}
		\advance \pgf@xa by -\pgfutil@tempdimb
		\advance \pgf@ya by \pgfutil@tempdima
		\pgfpathlineto{\pgfpoint{\pgf@xa}{\pgf@ya}}	     	
		\advance \pgf@xa by -\pgfutil@tempdimb
		\advance \pgf@ya by -\pgfutil@tempdima
		\pgfpathlineto{\pgfpoint{\pgf@xa}{\pgf@ya}}      	
		\advance \pgf@xa by \pgfutil@tempdimb
		\advance \pgf@ya by -\pgfutil@tempdima
		\pgfpathlineto{\pgfpoint{\pgf@xa}{\pgf@ya}}     
		\pgfusepath{draw}
	}
}

\makeatother

\usepackage{pgfplots}
\usepackage{pgfkeys}
\pgfplotsset{compat=newest}
\pgfplotsset{plot coordinates/math parser=false}

\DeclareMathOperator{\sinc}{sinc}

\DeclareMathAlphabet{\mathbit}{OML}{cmr}{bx}{it}
\DeclareMathAlphabet{\mathsf}{OT1}{cmss}{m}{n}
\DeclareMathAlphabet{\mathbsf}{OT1}{cmss}{bx}{it}
\newcommand{\inset}[2]{\ensuremath{\in \left\{#1,\ldots,#2\right\}}}
\newcommand{\Real}[1]{\ensuremath{\Re\left\{#1\right\}}}
\newcommand{\Imag}[1]{\ensuremath{\Im\left\{#1\right\}}}
\newcommand\diff[1]{\ensuremath{\:\mathrm{d}#1}}
\newcommand\deriv[2]{\ensuremath{\frac{\mathrm{d}#1}{\mathrm{d}#2}}}
\newcommand\pderiv[2]{\ensuremath{\frac{\partial#1}{\partial#2}}}
\newcommand\pderivk[3]{\ensuremath{\frac{\partial^{#3}#1}{\partial{#2}^{#3}}}}
\newcommand{\Deltarm}{\ensuremath{\mathrm{\Delta}}}
\usepackage{etoolbox}
\makeatletter
\patchcmd{\@maketitle}
{\addvspace{0.5\baselineskip}\egroup}
{\addvspace{-0.5\baselineskip}\egroup}
{}
{}

\def\ps@IEEEtitlepagestyle{%
	\def\@oddfoot{\mycopyrightnotice}%
	\def\@evenfoot{}%
}
\def\mycopyrightnotice{%
	{\begin{minipage}{\textwidth}\centering\footnotesize This work has been accepted for publication in IEEE/OSA Journal of Lightwave Technology. Digital Object Identifier: 10.1109/JLT.2017.2785316 \\ Copyright 0733-8724 \textcopyright 2017 IEEE. Personal use of this material is permitted. Permission from IEEE must be obtained for all other uses, in any current or future media, including reprinting/republishing this material for advertising or promotional purposes, creating new collective works, for resale or redistribution to servers or lists, or reuse of any copyrighted component of this work in other works.	\end{minipage}}% <--- Change here
	\gdef\mycopyrightnotice{}% just in case
}

\makeatother

\begin{document}
\title{Energy Conservation in Optical Fibers with Distributed Brick-Walls Filters}

\author{\IEEEauthorblockN{Javier Garc\'ia, Hassan Ghozlan, \IEEEmembership{Member,~IEEE}, and Gerhard Kramer, \IEEEmembership{Fellow,~IEEE}}
	
\IEEEcompsocitemizethanks{
\IEEEcompsocthanksitem Date of current version \today. J. Garc\'ia was supported by the German Research Foundation under Grant KR 3517/8-1. H. Ghozlan was supported in part by a USC Annenberg Fellowship and in part by the National Science Foundation under Grant CCF- 09-05235. G. Kramer was supported by an Alexander von Humboldt Professorship endowed by the German Federal Ministry of Education and Research.	
\IEEEcompsocthanksitem J. Garc\'ia and G. Kramer are with the Institute for Communications Engineering (LNT), Technical University of Munich, 80333 Munich, Germany.
\IEEEcompsocthanksitem H. Ghozlan was with the Department of Electrical Engineering, University of Southern California, Los Angeles, CA 90089, USA. He is now with Intel Corporation, Hillsboro, OR 97124, USA.
}
%\and
%\IEEEauthorblockN{Homer Simpson}
%\IEEEauthorblockA{Twentieth Century Fox\\
%\and
%\IEEEauthorblockN{James Kirk\\ and Montgomery Scott}
%\IEEEauthorblockA{Starfleet Academy\\
%San Francisco, California 96678-2391\\
%Fax: (888) 555--1212}
}


%Telephone: (800) 555--1212, Fax: (888) 555--1212}


\maketitle

\begin{abstract}
A band-pass filtering scheme is proposed to mitigate spectral broadening and channel coupling in the Nonlinear Schr\"odinger (NLS) fiber optic channel. The scheme is modeled by modifying the NLS Equation to include an attenuation profile with multiple brick-wall filters centered at different frequencies. It is shown that this brick-walls profile conserves the total in-band energy of the launch signal. Furthermore, energy fluctuations between the filtered channels are characterized, and conditions on the channel spacings are derived that ensure energy conservation in each channel. The maximum spectral efficiency of such a system is derived, and a constructive rule for achieving it using Sidon sequences is provided.
 
\end{abstract}

\begin{IEEEkeywords}
	Optical Kerr effect, Optical coupling, Energy conservation, Bandlimited communication
\end{IEEEkeywords}

% \fi
\IEEEpeerreviewmaketitle

\section{Introduction}\label{sec:intro}

The Nonlinear Schr\"odinger Equation (NLSE), which governs waveform propagation in fiber optic channels, exhibits signal-noise mixing and channel coupling %
that limit the achievable communication rates~\cite{essiambre_capacity}. In this work, we consider a modified NLSE with a distributed brick-walls filter that controls spectral broadening. In practical terms, this channel is a limiting case of a system that includes band-pass filters at regular spacings along the fiber, and when the filter spacing tends to zero. % Our experiments suggest that practically realizable filter spacings of $100$ km seem to be enough to make the model accurate at power values of about $1$ mW, which is already beyond the data rate peak of current systems. The effect could also be achieved by designing a fiber with very high attenuation outside the band of interest.
We prove that this band-limited system is energy-preserving. We further characterize the energy fluctuations between the filtered channels and show that avoiding Four Wave Mixing (FWM) conserves the per-channel energy. More precisely, we place the channels using Sidon sequences, as done in~\cite{forghieri_nofwm}, and we derive an asymptotic upper bound on the spectral efficiency.

The per-channel energy conservation suggests that filters placed at close spacings, or fibers with high attenuation outside the bands of interest, might increase communication rates at high powers. In fact, the lack of interference might allow pulse amplitude modulation rates to grow with launch power, avoiding the peak reported in~\cite{essiambre_capacity}. However, our results do not prove this because we measure energy by integrating over all time, and this prevents us from characterizing temporal broadening due to dispersion. We currently lack the analysis tools for analyzing energy over finite time intervals.

This paper is organized as follows. Section~\ref{sec:channel} describes the channel model and the brick-walls filters. Section~\ref{sec:conservation} proves total energy conservation, both in the frequency and time domains. Section~\ref{sec:coupling} characterizes the energy fluctuations between channels and gives a constructive rule for designing systems with per-channel energy conservation. Section~\ref{sec:conclusion} concludes the paper and gives directions for future work.

\section{System Model}\label{sec:channel}
Consider the slowly varying component $q(z, t)$ of an electrical field propagating along an optical fiber, where $z$ is distance and $t$ is time. The field obeys the Nonlinear Schr\"odinger equation (NLSE), which is expressed as~\cite[Eq. (2.3.46)]{agrawal_nfo}:
\begin{align}
\pderiv{}{z}q(z, t)=&-\frac{\alpha}{2} q(z, t)-j\frac{\beta_2}{2}\pderivk{}{t}{2}q(z, t) \nonumber\\
&+j\gamma\left|q(z, t)\right|^2 q(z, t)
\label{eq:nlse}
\end{align}
where $\alpha$ is the attenuation coefficient, $\beta_2$ is the dispersion coefficient, and $\gamma$ is the nonlinear coefficient. Taking the Fourier transform of \eqref{eq:nlse}, we obtain
\begin{multline}
\pderiv{}{z}Q(z, \omega)=\left(-\frac{\alpha(\omega)}{2}+j\frac{\beta_2}{2}\omega^2\right)Q(z, \omega)+j\frac{\gamma}{4\pi^2}\\\cdot\int\limits_{-\infty}^{\infty}\int\limits_{-\infty}^{\infty}Q(z, \omega_1)Q^*(z, \omega_2)Q(z, \omega-\omega_1+\omega_2)\diff{\omega_2}\diff{\omega_1}
\label{eq:nlse_f}
\end{multline}
where $\omega$ is the angular frequency, $Q^*$ denotes the complex conjugate of $Q$, and
\begin{equation}
Q(z, \omega)=\int\limits_{-\infty}^{\infty}q(z, t)e^{-j\omega t}\diff{t}.
\end{equation}
To obtain~\eqref{eq:nlse_f}, we used the standard result that a product of signals in the time domain becomes a convolution in the frequency domain. The convolution of two signals is
\begin{equation}
f(t)* g(t)= \int\limits_{-\infty}^{\infty} f(\tau)g(t-\tau)\diff\tau.
\end{equation}
%In the following, unless otherwise stated, the convolution variable is the one shared by the two operands.

The nonlinear term $\gamma$ causes variations in the spectral occupancy of signals. In a system with band-pass filters at discrete positions along the fiber, every filter removes part of the energy of the signal. One idea to mitigate the loss is to make the spacing $\Deltarm z$ between the filters small. We prove in Section~\ref{sec:conservation} that, in the limit when $\Deltarm z\to 0$,  \textit{the energy in the passband  is preserved as the signal propagates along the fiber}. An intuitive explanation is that, in a system with lumped filters at spacings $\Deltarm z$, the out-of-band energy produced by the nonlinearity grows in proportion to $\left(\Deltarm z\right)^2\;[\mathrm{J}]$. The energy loss per unit distance is therefore linear in $\Deltarm z\;[\mathrm{J}/\mathrm{m}]$, and goes to 0 as $\Deltarm z\to 0$. 

In~\eqref{eq:nlse_f}, we have allowed for a frequency-dependent attenuation profile $\alpha(\omega)$.  One way to model the distributed filtering is to design $\alpha(\omega)$ to be small in the bands of interest, and very large outside them. In the ideal case, we have
\begin{equation}
\alpha(\omega)=\begin{cases}
\alpha_0, & \omega\in\mathcal{S};\\
\infty, & \mathrm{otherwise}
\end{cases}
\label{eq:alpha}
\end{equation}
where $\mathcal{S}$ denotes a specified set of angular frequencies. We call~\eqref{eq:alpha} a \textit{brick-walls attenuation profile}. The corresponding \textit{brick-walls filter}  for $\alpha_0=0$ is
\begin{equation}
H_{\mathcal{S}}(\omega)=\begin{cases}
1, & \omega\in\mathcal{S};\\
0, & \mathrm{otherwise}.
\end{cases}
\label{eq:H}
\end{equation}
For example,  consider a generic WDM system with $N$ channels. For each channel $\mathcal{W}_n$, $n\inset{1}{N}$, let $\overline{\omega}_n$ be its center frequency, and $W_n$ its width. The band of interest is
\begin{equation}
\mathcal{W}=\bigcup\limits_{n=1}^{N} \mathcal{W}_n
\label{eq:W}
\end{equation}
where
\begin{equation}
\mathcal{W}_n=\left\{\omega : \left|\omega-\overline{\omega}_n\right|\le \frac{W_n}{2} \right\}
\end{equation}
and
\begin{equation}
\mathcal{W}_m\cap\mathcal{W}_n=\emptyset,\quad\mathrm{if}\ m\neq n.
\end{equation}
The channels are thus pairwise disjoint frequency intervals. The brick-walls filter $H(\omega)\triangleq H_{\mathcal{W}}(\omega)$ for $\alpha_0=0$ has an impulse response given by
\begin{equation}
h(t)=\sum\limits_{n=1}^{N}\frac{W_n}{2\pi}e^{j\overline{\omega}_n t}\mathrm{sinc}\left(\frac{W_n}{2\pi}t\right)
\label{eq:h}
\end{equation}
where $\mathrm{sinc}(x)\triangleq \sin(\pi x)/(\pi x)$.
\vspace{-0.45em}
%We prove that this design conserves total energy in Section~\ref{sec:conservation}, and we characterize the energy fluctuations between channels in Section~\ref{sec:coupling}.

\section{Energy conservation}\label{sec:conservation}

We prove energy conservation in the frequency domain in Section~\ref{sec:conservation_freq}, and in the time domain in Section~\ref{sec:conservation_time}. We write the energy in a frequency band $\mathcal{S}$ as
\begin{equation}
E_{\mathcal{S}}(z)\triangleq \frac{1}{2\pi}\int_{\mathcal{S}}\left|Q(z, \omega)\right|^2\diff{\omega}.
\label{eq:energy_band}
\end{equation}
We write $q(z, t)\in L^2$ (or $Q(z, \omega)\in L^2$) if the signal energy $E(z)=E_{\mathbb{R}}(z)$ exists in the Lebesgue sense and $E(z)<\infty$. Let $\Real{x}$ and $\Imag{x}$ be the real and imaginary parts of $x$, respectively. 

\subsection{Energy Conservation in the Frequency Domain}\label{sec:conservation_freq} 

We characterize the energy evolution for a bounded $\alpha(\omega)$.
\begin{lemma}
	\label{th:E_total_f}
	Let $Q\left(z, \omega\right)\in L^2$ be a signal propagating according to~\eqref{eq:nlse_f}, with attenuation profile $\alpha(\omega)<\infty$. The signal energy $E(z)$ evolves in $z$ according to
	\begin{equation}
	\deriv{}{z}E(z)\triangleq-\frac{1}{2\pi}\int\limits_{-\infty}^{\infty}\alpha(\omega)\left|Q(z, \omega)\right|^2\diff{\omega}.
	\label{eq:energy_evolution}
	\end{equation}
\end{lemma}

\begin{proof}
	Multiplying~\eqref{eq:nlse_f} by $Q^*(z, \omega)/(2\pi)$ and integrating over $\omega$, we obtain
	\begin{multline}
	\frac{1}{2\pi}\int\limits_{-\infty}^{\infty}Q^*(z, \omega)\pderiv{}{z}Q(z, \omega)\diff\omega\\=\frac{1}{2\pi}\int\limits_{-\infty}^{\infty}\left(-\frac{\alpha(\omega)}{2}+j\frac{\beta_2}{2}\omega^2\right)\left|Q(z, \omega)\right|^2\diff\omega+j\frac{\gamma}{\pi}I_{\mathbb{R}}(z)
	\label{eq:dEdz_int}
	\end{multline}
	where, for a set $\mathcal{S}$, we defined
	\begin{align}
	I_{\mathcal{S}}(z) & = & & \frac{1}{4\pi^2}\int_{\mathcal{S}}\int_{-\infty}^{\infty}\int_{-\infty}^{\infty}Q(z, \omega_1)Q^*(z, \omega_2) \nonumber \\ & & & \cdot Q(z, \omega-\omega_1+\omega_2)Q^*(z, \omega)\diff{\omega_1}\diff{\omega_2}\diff{\omega} \nonumber \\
	& = & & \frac{1}{4\pi^2}\int_{-\infty}^{\infty}\left[Q(z, \omega_3)*Q(z, \omega_3)\right] \nonumber \\ & & & \cdot\left[Q_{\mathcal{S}}(z, \omega_3)*Q(z, \omega_3)\right]^*\diff{\omega_3}
	\end{align}
	where $*$ denotes convolution in $\omega_3$. We used $\omega_3=\omega+\omega_2$, and defined $Q_{\mathcal{S}}(z, \omega)=H_{\mathcal{S}}(\omega)Q(z, \omega)$. Taking real parts of~\eqref{eq:dEdz_int}, and using $\Imag{I_{\mathbb{R}}(z)}=0$, we obtain~\eqref{eq:energy_evolution}.
\end{proof}

\begin{theorem}
	\label{th:energy_conservation}
	Let $Q\left(z, \omega\right)\in L^2$ be the Fourier transform of a bounded signal propagating according to~\eqref{eq:nlse_f}, with attenuation profile~\eqref{eq:alpha} and $\alpha_0\ge 0$. Then the signal has no energy outside the band $\mathcal{W}$ for $z>0$:
	\begin{equation}
	E_{\overline{\mathcal{W}}}(z)=0,\quad z>0
	\label{eq:bandlimited}
	\end{equation}
	where $\overline{\mathcal{W}}$ is the complement of $\mathcal{W}$ in $\mathbb{R}$. Furthermore, the signal energy $E(z)$ evolves in $z$ according to
	\begin{equation}
	E(z)=E_{\mathcal{W}}(0)e^{-\alpha_0 z}.
	\label{eq:energy_evolution_1}
	\end{equation}
\end{theorem}

\begin{proof}
	Let $\alpha_1>0$ and consider the attenuation profile
	\begin{equation}
	\alpha(\omega)=\begin{cases}
	\alpha_0, & \omega\in\mathcal{W} \\
	\alpha_1, & \mathrm{otherwise}.
	\end{cases}
	\label{eq:alpha_finite}
	\end{equation}
	
	Substituting~\eqref{eq:alpha_finite} in~\eqref{eq:energy_evolution}, we have
	\begin{equation}
	\deriv{}{z}E(z)=-\alpha_0 E_{\mathcal{W}}(z)-\alpha_1 E_{\overline{\mathcal{W}}}(z).
	\label{eq:dEdz_full}
	\end{equation}
%	This implies that $\mathrm{d}E(z)/\mathrm{d}z\le 0$, and $E(z)\le E(0)$ for $z>0$.
	
	Multiplying ~\eqref{eq:nlse_f} by $Q^*(z, \omega)/(2\pi)$, integrating over $\omega\in\overline{\mathcal{W}}$ and taking real part, we obtain
	\begin{equation}
	\deriv{}{z}E_{\overline{\mathcal{W}}}(z)=-\alpha_1 E_{\overline{\mathcal{W}}}(z)-\frac{\gamma}{\pi}\Im\left\{I_{\overline{\mathcal{W}}}(z)\right\}.
	\label{eq:dEdz_out}
	\end{equation}
	
	%	We have $\Imag{I_{\mathbb{R}}(z)}=0$, and therefore $\Imag{I_{\overline{\mathcal{W}}}(z)}=-\Imag{I_{\mathcal{W}}(z)}$. Using this in~\eqref{eq:dEdz_out} we have
	%	\begin{equation}
	%	\deriv{}{z}E_{\overline{\mathcal{W}}}(z)=-\alpha_1 E_{\overline{\mathcal{W}}}(z)+2\gamma\Im\left\{I_{\mathcal{W}}(z)\right\}
	%	\label{eq:dEdz_out2}
	%	\end{equation}
	By assumption, the time-domain signal $q(z, t)$ is bounded. Let $\left|q(z, t)\right|\le q_{\max}$. We have
	\begin{align}
	\left|I_{\overline{\mathcal{W}}}(z)\right|&=\left|\int_{-\infty}^{\infty}\left[H_{\overline{\mathcal{W}}}(\omega)Q^*(z, \omega)\right]\left[\frac{1}{4\pi^2}\int_{-\infty}^{\infty}\int_{-\infty}^{\infty}Q(z, \omega_1) \right. \right. \nonumber \\ &  \left. \left. \vphantom{\int_{-\infty}^{\infty}} \cdot Q^*(z, \omega_2)Q(z, \omega-\omega_1+\omega_2)\diff{\omega_1}\diff{\omega_2}\right]\diff{\omega}\right| \nonumber \\
	& \stackrel{\text{(a)}}{=}  2\pi\left|\int_{-\infty}^{\infty}q_{\overline{\mathcal{W}}}^*(z, t)\left|q(z, t)\right|^2q(z, t)\diff t\right| 
	\end{align}	
	where $q_{\overline{\mathcal{W}}}(z, t)$ is the inverse Fourier transform of $Q_{\overline{\mathcal{W}}}(z, \omega)$, and (a) follows from Plancherel's formula~\cite[Eq. (2.1.2)]{titchmarsh_plancherel}
	\begin{equation}
	2\pi\int_{-\infty}^{\infty} x(t)y^*(t)\diff t=\int_{-\infty}^{\infty}X(\omega)Y^*(\omega)\diff \omega.
	\end{equation}	
	Using the standard result $\left|\left\langle x, y\right\rangle\right|\le\left\|x\right\|\left\|y\right\|$, we obtain
	\begin{align}	
	\left|I_{\overline{\mathcal{W}}}(z)\right| & \stackrel{\text{(b)}}{\le}  2\pi\sqrt{\int_{-\infty}^{\infty}|q_{\overline{\mathcal{W}}}(z, t)|^2\diff t\int_{-\infty}^{\infty}|q(z, t)|^6\diff t } \nonumber \\
	&  \le 2\pi\sqrt{\int_{-\infty}^{\infty}|q_{\overline{\mathcal{W}}}(z, t)|^2\diff t\int_{-\infty}^{\infty}q_{\max}^4|q(z, t)|^2\diff t } \nonumber \\
	&  \le 2\pi q_{\max}^2\sqrt{E(0)}\sqrt{E_{\overline{\mathcal{W}}}(z)}.
	\label{eq:I_mod}
	\end{align}
	The last inequality follows from~\eqref{eq:dEdz_full} and $\alpha(\omega)\ge 0$, so that $E(z)\le E(0)$ for $z\ge 0$. Using~\eqref{eq:I_mod} and $-\Imag{x}\le|x|$, we upper bound~\eqref{eq:dEdz_out} as
	\begin{equation}
	\deriv{}{z}E_{\overline{\mathcal{W}}}(z)\le -\alpha_1E_{\overline{\mathcal{W}}}(z)+2|\gamma|q_{\max}^2\sqrt{E(0)}\sqrt{E_{\overline{\mathcal{W}}}(z)}.
	\end{equation}	
	The solutions to a differential inequality can be bounded by the solution to the corresponding equality~\cite[p. 7]{szarski_di}. This yields
	\begin{align}
	E_{\overline{\mathcal{W}}}(z) & \le E_{\overline{\mathcal{W}}}(0)e^{-\alpha_1 z} \nonumber \\ & +\frac{4|\gamma|q_{\max}^2\sqrt{E(0)}\sqrt{E_{\overline{\mathcal{W}}}(0)}}{\alpha_1}e^{-\frac{\alpha_1}{2}z}\left(1-e^{-\frac{\alpha_1}{2}z}\right)  \nonumber \\ & +\frac{4|\gamma|^2q_{\max}^4E(0)}{\alpha_1^2}\left(1-e^{-\frac{\alpha_1^2}{2}z}\right)^2.
	\label{eq:E_out_bound}
	\end{align}
	Letting $\alpha_1\to\infty$ proves~\eqref{eq:bandlimited}, and therefore $E(z)=E_{\mathcal{W}}(z)$ for $z>0$. From~\eqref{eq:E_out_bound}, we have $\lim_{\alpha_1\to\infty} \alpha_1E_{\overline{\mathcal{W}}}(z)=0$. Using this in~\eqref{eq:dEdz_full} proves~\eqref{eq:energy_evolution_1}.
\end{proof}

Theorem~\ref{th:energy_conservation} shows that if the launch signal has energy inside $\mathcal{W}$ only,  then there is no energy loss due to the filters. The exponential energy loss is due to attenuation only.% For $\alpha_0=0$, we have $E(z)=E_{\mathcal{W}}(0)$.

\subsection{Energy Conservation in the Time Domain}\label{sec:conservation_time}
Consider next the time domain approach. Adding the frequency dependence of $\alpha$ to~\eqref{eq:nlse}, we obtain
\begin{align}
\pderiv{}{z}q(z, t)=&-\frac{1}{2}a(t) * q(z, t)-j\frac{\beta_2}{2}\pderivk{}{t}{2}q(z, t) \nonumber\\
&+j\gamma\left|q(z, t)\right|^2 q(z, t)
\label{eq:nlse_t}
\end{align}
where $a(t)$ is the inverse Fourier transform of $\alpha(\omega)$.

\begin{theorem}
	Let $q(z, t)\in L^2$ be a bounded signal propagating according to~\eqref{eq:nlse_t}, with brick-walls attenuation profile~\eqref{eq:alpha}. Then the signal energy
	\begin{equation}
	E(z)=\int\limits_{-\infty}^{\infty}\left|h(t)*q(z, t)\right|^2\diff t
	\end{equation}
	evolves in $z$ according to
	\begin{equation}
	E(z)=E_{\mathcal{W}}(0)e^{-\alpha_0 z},\quad z>0.
	\label{eq:energy_exp_2}
	\end{equation}
	\label{th:E_time}
\end{theorem}
\begin{proof}
	We use a split-step approximation. We write~\eqref{eq:nlse_t} as
	\begin{equation}
	\pderiv{}{z}q(z, t)=(\hat{D}+\hat{A}+\hat{N})q(z, t)
	\label{eq:q_z_op}
	\end{equation}
	where $\hat{N}=j\gamma\left|q(z, t)\right|^2$, $\hat{A}=-\frac{1}{2}a(t)*$, and $\hat{D}=-j\frac{\beta_2}{2}\pderivk{}{t}{2}$ are the nonlinearity, attenuation and dispersion operators, respectively. Consider a small step $\Deltarm z$. The solution of~\eqref{eq:q_z_op} is
	\begin{align}
	q(z+\Deltarm z, t)&=e^{\Deltarm z(\hat{D}+\hat{A}+\hat{N})}q(z, t) \nonumber \\ & \approx e^{\Deltarm z\hat{D}}e^{\Deltarm z\hat{A}}e^{\Deltarm z\hat{N}}q(z, t)
	\label{eq:split_step_exp}
	\end{align}
	where the approximation is valid up to first order in $\Deltarm z$~\cite[p. 48]{agrawal_nfo}. Consider the dispersion step:
	
	%\begin{figure}[tbp]\centering
	%	\begin{tikzpicture}[
	%		node distance=1.5em,
	%		block/.style={rectangle,
	%			minimum width=3em,
	%			minimum height=1.36em,
	%			draw}
	%	]
	%	\node (start) {$q(z, t)$};
	%	\node (D) [block, right=of start] {$e^{\Deltarm z \hat{D}}$};
	%	\node (N) [block, right=of D] {$e^{\Deltarm z \hat{N}}$};
	%	\node (A) [block, right=of N] {$e^{\Deltarm z \hat{A}}$};
	%	\node (end) [right=of A] {$q(z+\Deltarm z, t)$};
	%	
	%	\draw [-Stealth] (start) -- (D);
	%	\draw [-Stealth] (D) -- node [above] {$q_1$} (N);
	%	\draw [-Stealth] (N) -- (A);
	%	\draw [-Stealth] (A) -- (end);
	%	
	%	\end{tikzpicture}
	%	\caption{Split-step model of the NLSE}
	%	\label{fig:split_step}
	%\end{figure}
	
	%Figure~\ref{fig:split_step} depicts the idea of the split-step approximation. The dispersion $e^{\Deltarm z\hat{D}}$, the nonlinearity $e^{\Deltarm z\hat{N}}$ and the frequency-dependent attenuation $e^{\Deltarm z\hat{A}}$ are applied successively for a small $\Deltarm z$.
	\begin{equation}
	q(z+\Deltarm z, t)\triangleq\hat{F}^{-1}e^{j\Deltarm z \frac{\beta_2}{2}\omega^2}\hat{F}q_a(z, t)
	\end{equation}
	where $q_a(z, t)\triangleq e^{\Deltarm z\hat{A}}e^{\Deltarm z\hat{N}}q(z, t)$, and $\hat{F}$ applies a Fourier transform. This step causes only phase shifts in the frequency domain which conserve the energy:
	\begin{equation}
	E(z+\Deltarm z)=\int\limits_{-\infty}^{\infty}\left|q(z+\Deltarm z, t)\right|^2\diff t=\int\limits_{-\infty}^{\infty}\left|q_a(z, t)\right|^2\diff t.
	\end{equation}
	
	Consider now the nonlinearity and attenuation steps:
	\begin{equation}
	q_a(z, t)=e^{-\frac{\alpha_0}{2}\Deltarm z}h(t)*\left[q(z, t)e^{j\gamma \Deltarm z \left|q(z, t)\right|^2}\right]
	\label{eq:s}
	\end{equation}
	and
	\begin{multline}
	\int\limits_{-\infty}^{\infty}\left|q_a(z, t)\right|^2\diff t= \int\limits_{-\infty}^{\infty}\int\limits_{-\infty}^{\infty}\int\limits_{-\infty}^{\infty}q(z, t_1)q^*(z, t_2) \\ \cdot e^{\Deltarm z\left[j\gamma\left(\left|q(z, t_1)\right|^2-\left|q(z, t_2)\right|^2\right)\right]} \\ \cdot e^{-\Deltarm z \alpha_0}h(t-t_1)h^*(t-t_2) \diff t_2 \diff t_1 \diff t.
	\end{multline}
	As $q(z, t)$ is band-limited to $\mathcal{W}$ for $z>0$, we have
	\begin{equation}
	E(z)=\int\limits_{-\infty}^{\infty}\left|q(z, t)\right|^2\diff t.
	\end{equation}
	The derivative of the energy is
	\begin{align}
	%\begin{array}{>{\displaystyle}{r}>{\displaystyle}{c}>{\displaystyle}{l}}
	\deriv{}{z}E(z) & &=&\lim\limits_{\Deltarm z\rightarrow 0}\frac{E(z+\Deltarm z)-E(z)}{\Deltarm z} \nonumber \\
	& &=&\lim\limits_{\Deltarm z\rightarrow 0} \int\limits_{-\infty}^{\infty}\int\limits_{-\infty}^{\infty}\int\limits_{-\infty}^{\infty} q(z, t_1)q^*(z, t_2) \nonumber \\ & & & \cdot\frac{e^{\Deltarm z\left[-\alpha_0+j\gamma\left(\left|q(z, t_1)\right|^2-\left|q(z, t_2)\right|^2\right)\right]}-1}{\Deltarm z} \nonumber \\ & & & \cdot h(t-t_1)h^*(t-t_2) \diff t_2 \diff t_1 \diff t \nonumber \\ & &=& -\alpha_0 E(z)+I_1+I_1^* %\nonumber \\ & &=& -\alpha_0 E(z)+2\Re\left\{I_1\right\}
	%\end{array}
	\label{eq:dEdz_time_1}
	\end{align}
	where we used $\lim_{\Deltarm z\rightarrow 0} \left(e^{b\Deltarm z}-1\right)/\Deltarm z=b$ and
	\begin{multline}
	I_1=j\gamma\int\limits_{-\infty}^{\infty}\int\limits_{-\infty}^{\infty} q(z, t_1)\left|q(z, t_1)\right|^2h(t-t_1)\\ \cdot\left[\int\limits_{-\infty}^{\infty} q(z, t_2)h(t-t_2) \diff t_2 \right]^* \diff t_1 \diff t.
	\label{eq:I_1}
	\end{multline}
	
	In~\eqref{eq:dEdz_time_1} we used the boundedness of $q(z, t)$ to ensure that the integrand is bounded by an integrable function, and therefore the limit operator can be moved inside the integrals~\cite[Thm. 10.21]{browder2012mathematical}. The integral in square brackets in~\eqref{eq:I_1} is $h(t)*q(z, t)$. As $q(z, t)$ is inside the passband $\mathcal{W}$ of the brick-walls filter $h(t)$, the result of this convolution is exactly $q(z, t)$. Using $h(t)=h^*(-t)$, we obtain
	%\begin{multline}
	%\deriv{}{z}E_{\mathcal{W}}(z)=-\alpha_0 E_{\mathcal{W}}(z)-2\gamma\Im\left\{\int\limits_{-\infty}^{\infty}\int\limits_{-\infty}^{\infty} q_d(z, t_1)\left|q_d(z, t_1)\right|^2\right.\\ \left.\cdot h(t-t_1)s^*(t) \diff t_1 \diff t\vphantom{\int\limits_{-\infty}^{\infty}}\right\}.
	%\label{eq:dEdz_time_2}
	%\end{multline}
	%From~\eqref{eq:h}, it is easy to see that $h(-t)=h^*(t)$. This allows us to rewrite~\eqref{eq:dEdz_time_2} as:
	\begin{align}
	I_1 & =j\gamma\int\limits_{-\infty}^{\infty} q(z, t_1)\left|q(z, t_1)\right|^2\int\limits_{-\infty}^{\infty} h^*(t_1-t)q^*(z, t) \diff t \diff t_1 \nonumber \\ & = j\gamma\int\limits_{-\infty}^{\infty} \left|q(z, t_1)\right|^4 \diff t_1.
	\label{eq:dEdz_time_3}
	\end{align}
	Substituting~\eqref{eq:dEdz_time_3} into~\eqref{eq:dEdz_time_1} proves~\eqref{eq:energy_exp_2}.
\end{proof}

\begin{remark}
	The preservation of energy suggests that capacity should grow with power. However, the variable $E(z)$ is defined based on an infinite time interval, allowing only one channel use. We thus require a more precise result on the energy in a finite time interval to be able to make a statement about the channel capacity.
\end{remark}

\section{Energy fluctuations between channels}\label{sec:coupling}

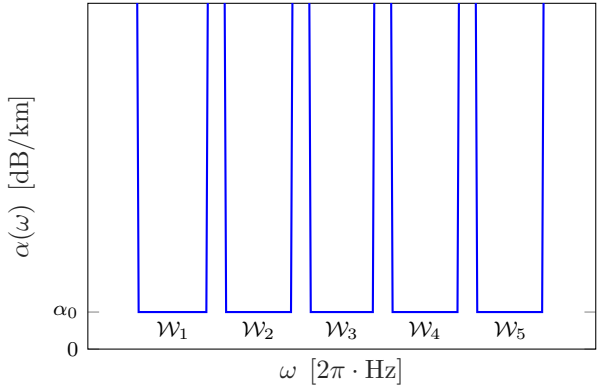
\begin{figure}[tbp]\centering
\begin{tikzpicture}

\begin{axis}[%
width=0.951\figurewidth,
height=\figureheight,
at={(0\figurewidth,0\figureheight)},
scale only axis,
xmin=-23.5619449019235,
xmax=23.5619449019235,
xtick={\empty},
xlabel style={font=\color{white!15!black}},
xlabel={$\omega\;\left[2\pi\cdot\mathrm{Hz}\right]$},
ymin=-0.06,
ymax=0.5,
ytick={-0.06, 0},
yticklabels={$0$, $\alpha_0$},
ylabel style={font=\color{white!15!black}},
ylabel={$\alpha(\omega)\;\left[\mathrm{dB}/\mathrm{km}\right]$},
axis background/.style={fill=white}
]
\addplot [color=blue, thick, forget plot]
  table[row sep=crcr]{%
-23.5619449019235	1\\
-19.1113553093379	1\\
-18.8495559215388	0\\
-12.5663706143592	0\\
-12.30457122656	1\\
-10.9955742875643	1\\
-10.7337748997651	0\\
-4.71238898038469	0\\
-4.45058959258554	1\\
-3.1415926535898	1\\
-2.87979326579065	0\\
2.87979326579065	0\\
3.1415926535898	1\\
4.45058959258554	1\\
4.71238898038469	0\\
10.7337748997651	0\\
10.9955742875643	1\\
12.30457122656	1\\
12.5663706143592	0\\
18.5877565337396	0\\
18.8495559215388	1\\
23.5619449019235	1\\
};
\node at (-15.65, 0) [anchor=north] {$\mathcal{W}_1$};
\node at (-7.72, 0) [anchor=north] {$\mathcal{W}_2$};
\node at (0, 0) [anchor=north] {$\mathcal{W}_3$};
\node at (7.72, 0) [anchor=north] {$\mathcal{W}_4$};
\node at (15.65, 0) [anchor=north] {$\mathcal{W}_5$};
\end{axis}
\end{tikzpicture}%
	\caption{Brick-walls attenuation profile.}
	\label{fig:alpha_w}
\end{figure}

\subsection{General Expression for Energy Fluctuations}
We have so far established energy conservation in a system with brick-walls attenuation profile~\eqref{eq:alpha} and $\alpha_0=0$. Figure~\ref{fig:alpha_w} shows $\alpha(\omega)$ for a $5$-channel system, where the attenuation is $\alpha_0$ in the bands $\mathcal{W}_n$, and $\infty$ outside them. We now characterize the energy \textit{per WDM channel}. Consider the signal
\begin{equation}
Q_n(z, \omega)=H_n(\omega)Q(z, \omega)
\label{eq:Q_n}
\end{equation}
in channel $\mathcal{W}_n$, where $H_n(\omega)\triangleq H_{\mathcal{W}_n}(\omega)$.
%\begin{equation}
%H_n(\omega)=\begin{cases}
%1,& \left|\omega-\overline{\omega}_n\right|\le\frac{W_n}{2} \\
%0,& \mathrm{otherwise}.
%\end{cases}
%\end{equation}

\begin{theorem}\label{th:energy_coupling}
	Let $Q(z, \omega)\in L^2$ be a bounded frequency-domain signal propagating according to~\eqref{eq:nlse_f} with the brick-walls attenuation profile~\eqref{eq:alpha}. Then the energy $E_{\mathcal{W}_n}(z)$ in channel $\mathcal{W}_n, n\inset{1}{N}$, evolves in $z$ according to
	\begin{multline}
	\deriv{}{z}E_{\mathcal{W}_n}(z)=-\alpha_0 E_{\mathcal{W}_n}(z)-\frac{\gamma}{4\pi^3}\Im\left\{ \vphantom{\int\limits_{-\infty}^{\infty}} \right. \\ \left. \int\limits_{-\infty}^{\infty}\left[Q(z, \omega_3)* Q(z, \omega_3)\right]\cdot\left[Q_n(z, \omega_3)* Q(z, \omega_3)\right]^*\diff\omega_3\vphantom{\int\limits_{-\infty}^{\infty}}\right\}
	\label{eq:energy_evolution_channel}
	\end{multline}
	where the convolutions are in $\omega_3$.
\end{theorem}

\begin{proof}
	The proof follows the same steps as the proof of Theorem~\ref{th:E_total_f}.%, but using the new integration set $\omega\in \mathcal{W}_n$.
	%
%	The proof is similar to that of Theorem~\ref{th:E_total_f}, with appropriately modified limits of integration. This time, we integrate in $\omega$ over $\mathcal{W}_n$ instead of $\mathcal{W}$, which changes the integration limits of~\eqref{eq:dEdz_conv} to:
%	\begin{multline}
%	\deriv{}{z}\left\{\frac{1}{2\pi}\int\limits_{\mathcal{W}_n}\left|Q(z, \omega)\right|^2\diff\omega\right\}\\=-\frac{1}{2\pi}\int\limits_{\mathcal{W}_n}\alpha(\omega)\left|Q(z, \omega)\right|^2\diff\omega\\-\frac{\gamma}{4\pi^3}\Im\left\{\int\limits_{-\infty}^{\infty}\left[\int\limits_{-\infty}^{\infty}Q(z, \omega_1)Q(z, \omega_3-\omega_1)\diff{\omega_1}\right]\right.\\ \left.\cdot\left[\int\limits_{\mathcal{W}_n}Q^*(z, \omega)Q^*(z, \omega_3-\omega)\diff\omega\right]\diff\omega_3\right\}
%	\label{eq:dEdz_ch_int}
%	\end{multline}
%	which is equivalent to~\eqref{eq:energy_evolution_channel}.	
\end{proof}

Theorem~\ref{th:energy_coupling} implies that the energy per channel is \textit{not necessarily conserved}, even if $\alpha_0=0$. 
\subsection{Example: Three-Tone System}
To illustrate the implications of Theorem~\ref{th:energy_coupling}, suppose we use $N=3$ vanishingly thin channels (tones) with spacing $\Deltarm\omega$:
\begin{equation}
\alpha_0=0;\ \overline{\omega}_1=-\Deltarm\omega;\ \overline{\omega}_2=0;\ \overline{\omega}_3=\Deltarm\omega;\ W_n\rightarrow 0.
\end{equation}
The three-tone system is the limit of three modulated rectangular pulses when their duration $T$ goes to infinity:
\begin{align}
Q(z, \omega)= & \lim_{T\rightarrow\infty} 2\pi Q_1(z)\delta_T(\omega+\Deltarm\omega)+2 \pi Q_2(z)\delta_T(\omega) \nonumber \\ & +2\pi Q_3(z)\delta_T(\omega-\Deltarm\omega)
\label{eq:threetone}
\end{align}
where
\begin{equation}
\delta_T(\omega)=\frac{T}{2\pi}\sinc\left(\frac{T}{2\pi}\omega\right).
\end{equation}
We have $\delta_T(\omega)*\delta_T(\omega)=\delta_T(\omega)$, $\int_{-\infty}^{\infty}\left|\delta_T(\omega)\right|^2\diff\omega=T/(2\pi)$, and we compute the power of the tones to be
\begin{equation}
P_{\mathcal{W}_n}(z)=\lim_{T\rightarrow\infty} \frac{E_{\mathcal{W}_n}(z)}{T}=\left|Q_n(z)\right|^2.
\end{equation}
Substituting~\eqref{eq:threetone} in~\eqref{eq:nlse_f}, we obtain the FWM equations for the evolution of each tone. For $\omega_1$ we have
\begin{multline}
\deriv{}{z}Q_1(z)=-j\frac{\beta_2}{2}Q_1(z)+j\gamma\left[\left|Q_1(z)\right|^2Q_1(z)\right. \\  \left.+2\left(\left|Q_2(z)\right|^2+\left|Q_3(z)\right|^2\right)Q_1(z)+Q_1(z)Q_2^*(z)Q_3(z)\right].
\label{eq:fwm}
\end{multline}
Multiplying~\eqref{eq:fwm} by $Q_1^*(z)$ and taking real parts, we have
\begin{align}
\deriv{}{z}P_{\mathcal{W}_1}(z)&=-2\gamma\Im\left\{\left|Q_1(z)\right|^4+2\left|Q_2(z)\right|^2\left|Q_1(z)\right|^2\right.\nonumber\\ &\left.+2\left|Q_3(z)\right|^2\left|Q_1(z)\right|^2+Q_1^*(z)Q_2(z)^2Q_3^*(z)\right\}\nonumber\\&=-2\gamma\Imag{Q_1^*(z)Q_2(z)^2Q_3^*(z)}.
\label{eq:energy_threedeltas_1}
\end{align}
Only the FWM term $Q_1^*(z)Q_2(z)^2Q_3^*(z)$ affects the signals' powers. The self-phase modulation (SPM) term $\left|Q_1(z)\right|^4$ and the cross-phase modulation (XPM) terms $2\left|Q_2(z)\right|^2\left|Q_1(z)\right|^2$ and $2\left|Q_3(z)\right|^2\left|Q_1(z)\right|^2$ affect only the signals' phases. The power evolutions for the other channels are
\begin{align}
\deriv{}{z}P_{\mathcal{W}_2}(z)&=-4\gamma\Imag{Q_1(z)Q_2^{*}(z)^2 Q_3(z)}\nonumber \\
\deriv{}{z}P_{\mathcal{W}_3}(z)&=-2\gamma\Imag{Q_1^*(z)Q_2(z)^2Q_3^*(z)}.
\end{align}
We see that the total power of the system is conserved:
\begin{equation}
\mathrm{d}\left[P_{\mathcal{W}_1}(z)+P_{\mathcal{W}_2}(z)+P_{\mathcal{W}_3}(z)\right]/\mathrm{d}z=0.
\end{equation}

%Figure~\ref{fig:energy_3tone} shows the evolution of the powers in a three-tone system with the parameters in Table~\ref{tab:parameters}. The frequency spacing is $\Deltarm\omega=402.12\cdot 10^{9}\ \mathrm{rad/s}$. The launch amplitudes are
%\begin{align}
%Q_1(0)=Q_2(0)&=0.5274\;\sqrt{\mathrm{Watt}} \nonumber \\
%Q_3(0)&=0.5015\;\sqrt{\mathrm{Watt}}
%\end{align}
%The launch phases are all $0$. The split-step Fourier method~\cite{agrawal_nfo} was used to simulate the NLSE with distributed filtering. A time window of $2\;\mathrm{ns}$ was used for the simulations.
\begin{remark}
	In a two-tone system, the power in each channel is conserved. This can be shown by setting $Q_3(z)=0$. 
\end{remark}

%\begin{figure}[tbp]\centering
%	\setlength\figurewidth{0.43\textwidth}
%	\setlength\figureheight{0.648\figurewidth}
%	\input{Graphs/energy_3tone.tex}
%	\caption{Power evolution in a $3$-tone system with distributed filters.}
%	\label{fig:energy_3tone}
%\end{figure}
%

\subsection{Conditions for Per-Channel Energy Conservation}
\begin{definition}
	A system with attenuation profile~\eqref{eq:alpha} is \textup{energy-decoupled} if, for any bounded launch signal $Q(0, \omega)\in L^2$ and for each channel $n$, the evolution of its energy $E_{\mathcal{W}_n}(z)$ is the same as if only that channel were present in the system:
	\begin{equation}
	E_{\mathcal{W}_n}(z)=E_{\mathcal{W}_n}(0) e^{-\alpha_0 z}\quad \forall Q(0, \omega), \forall n.
	\label{eq:no_coupling}
	\end{equation}
	Otherwise, the system is said to be \textup{energy coupled}.
\end{definition}
Equation~\eqref{eq:energy_threedeltas_1} suggests that energy coupling is caused by the FWM terms. This motivates us to try to avoid FWM in order to ensure energy decoupling, as also done in~\cite{forghieri_nofwm}. We use the standard definition of the sum of sets:
%\begin{align}
%\mathcal{W}_{n_1}+\mathcal{W}_{n_2}&=\left\{\omega_{n_1}+\omega_{n_2}:\omega_{n_1}\in\mathcal{W}_{n_1}, \omega_{n_2}\in\mathcal{W}_{n_2} \right\} \nonumber \\
%&=\left[\omega_{n_1}^{(\ell)}+\omega_{n_2}^{(\ell)},\ \omega_{n_1}^{(\mathrm{r})}+\omega_{n_2}^{(\mathrm{r})}\right],
%\end{align}
\begin{equation}
\mathcal{W}_{n_1}+\mathcal{W}_{n_2}=\left\{\omega_{n_1}+\omega_{n_2}:\omega_{n_1}\in\mathcal{W}_{n_1}, \omega_{n_2}\in\mathcal{W}_{n_2} \right\}.
\end{equation}
%where $\omega_{n}^{(\ell)}\triangleq\overline{\omega}_n-W_n/2$ and $\omega_{n}^{(\mathrm{r})}\triangleq\overline{\omega}_n+W_n/2$ are respectively the left and right boundaries of the interval $\mathcal{W}_n$.

\begin{lemma}
	Let $\mathcal{W}$ be a multi-channel frequency band as defined in~\eqref{eq:W}. Then the NLSE system~\eqref{eq:nlse_f} with attenuation profile~\eqref{eq:alpha} is energy-decoupled~\eqref{eq:no_coupling} if and only if
	\begin{equation}
	\left(\mathcal{W}_{n_1}+\mathcal{W}_{n_2}\right)\cap \left(\mathcal{W}_{n}+\mathcal{W}_{n_3}\right)=\emptyset,\ \forall \left\{n_1,n_2\right\}\neq\left\{n, n_3\right\}.
	\label{eq:condition}
	\end{equation}
	The notation $\forall\left\{n_1, n_2\right\}\neq\left\{n, n_3\right\}$ allows any $n_1$, $n_2$, $n$ and $n_3$ in $\left\{1,\ldots,N\right\}$ such that $(n_1, n_2)\ne(n, n_3)$ and $(n_1, n_2)\ne(n_3, n)$, where $(a, b)$ denotes an ordered pair. For example,
	\begin{equation}
	(n_1, n_2)=(1, 1),\quad (n, n_3)=(1, 2)
	\end{equation}
	is allowed, but
	\begin{equation}
	(n_1, n_2)=(1, 2),\quad (n, n_3)=(2, 1)
	\end{equation}
	is not.
	
\end{lemma}

\begin{proof}
Substituting $Q(z, \omega)=\sum_{n'=1}^{N} Q_{n'}(z,\omega)$ in~\eqref{eq:energy_evolution_channel}, the condition to prevent energy coupling is
\begin{equation}
\Im\left\{\sum_{n_1=1}^{N}\sum_{n_2=1}^{N}\sum_{n_3=1}^{N}I_{n_1n_2nn_3}\right\}=0,\quad\forall n\inset{1}{N}
\label{eq:condition_sum_simpl}
\end{equation}
where
\begin{multline}
I_{n_1n_2nn_3}\triangleq \int_{-\infty}^{\infty}\left[Q_{n_1}(z,\omega)*Q_{n_2}(z, \omega)\right] \\ \cdot\left[Q_n(z, \omega)*Q_{n_3}(z,\omega)\right]^*\diff\omega.
\end{multline}
The terms with $\left\{n_1, n_2\right\}=\left\{n, n_3\right\}$ are real-valued. If~\eqref{eq:condition} is true, then all the other terms are products of convolutions that do not overlap in frequency, which implies~\eqref{eq:no_coupling}.

To prove that~\eqref{eq:no_coupling} implies~\eqref{eq:condition}, suppose that~\eqref{eq:condition} is false, i.e., there is a choice of $\left\{m_1, m_2\right\}\neq\left\{m, m_3\right\}$ for which
\begin{equation}
\left(\mathcal{W}_{m_1}+\mathcal{W}_{m_2}\right)\cap \left(\mathcal{W}_{m}+\mathcal{W}_{m_3}\right)\neq\emptyset.
\end{equation}
Choose $\omega_{m_1}\in\mathcal{W}_{m_1}$, $\omega_{m_2}\in\mathcal{W}_{m_2}$, $\omega_{m}\in\mathcal{W}_{m}$, $\omega_{m_3}\in\mathcal{W}_{m_3}$, such that
\begin{equation}
\omega_{m_1}+\omega_{m_2}=\omega_m+\omega_{m_3}.
\end{equation}
Now consider $n=m$ in~\eqref{eq:condition_sum_simpl}, and choose the following four-tone launch signal:
\begin{equation}
Q_n(0, \omega)=\begin{cases}
\delta\left(\omega-\omega_n\right),& \substack{\mathrm{if}\ n\in\left\{m_1, m_2, m_3\right\}\\ \mathrm{and}\ n\neq m}\\
e^{-j\pi/8}\delta\left(\omega-\omega_n\right),& \mathrm{if}\ n=m \\
0, & \mathrm{otherwise}.
\end{cases}
\end{equation}
The sum~\eqref{eq:condition_sum_simpl} has $64$ terms. We call $I_{n_1n_2mn_3}$ an \textit{overlapping term} if $\omega_{n_1}+\omega_{n_2}=\omega_m+\omega_{n_3}$, and a \textit{non-overlapping term} otherwise. Non-overlapping terms are equal to $0$. We consider groups of terms $I_{n_1n_2mn_3}$ with nonnegative imaginary part.
\begin{itemize}
	\item \textit{Case 1) The $27$ terms of the form $I_{n_1n_2mn_3}$, where $n_1, n_2, n_3\neq m$}: these terms are equal to $\exp\left(j\pi/8\right)$ if they are overlapping, with positive imaginary part.
	\item \textit{Case 2) The $9$ terms of the form $I_{n_1n_2mm}$ where $n_1, n_2\neq m$ and $n_3=m$}: these terms are equal to $\exp\left(j\pi/4\right)$ if they are overlapping, with positive imaginary part.
	\item \textit{Case 3) The $18$ terms of the form $I_{mn_2mn_3}$ or $I_{n_1mmn_3}$ where $n_3\neq m$, and either $n_1=m$ or $n_2=m$}: these terms are equal to $1$ if they are overlapping.
	\item \textit{Case 4) The $9$ terms where only one of $n_1, n_2, n_3$ is not equal to $m$}: these terms form the following groups:
	\begin{equation*}
	I_{mmmn_3}+I_{mn_3mm}+I_{n_3mmm}=2I_{n_3mmm}+I_{n_3mmm}^*
	\end{equation*}
	where $I_{n_3mmm}=\exp\left(j\pi/8\right)$ if it is overlapping. The imaginary part is again positive.
	\item \textit{Case 5) The term $I_{mmmm}=1$}.
	
\end{itemize}
This list shows that all the terms of the sum~\eqref{eq:condition_sum_simpl} have either positive or zero imaginary part. If any term in Cases 1, 2 or 4 is overlapping, then the sum has positive imaginary part and the proof is complete. Case 5 is not included in~\eqref{eq:condition}. If \textit{only} terms of Case 3 are overlapping, then we have an overlapping term $I_{mn_2mn_3}$ with $n_2\neq n_3$ (if $n_2=n_3$, the term is again not included in~\eqref{eq:condition}). But then $I_{n_2mn_3m}$ is an overlapping term of Case 1, and thus~\eqref{eq:condition_sum_simpl} is not satisfied for $n=n_3$. This completes the proof: if any overlapping term exists, then the per-channel energy is not conserved. 
\end{proof}
\begin{remark}
	Condition~\eqref{eq:condition} in a system with distributed filtering ensures the absence of FWM, and thus ensures per-channel energy conservation. However, interaction between channels still occurs due to XPM.
\end{remark}

\subsection{Spectral Efficiency of an Energy-Decoupled System}\label{sec:max_efficiency}
Energy-decoupled systems have the potential to encode information in the energy of each channel and communicate without inter-channel interference. We develop an upper bound on the \textit{spectral filling efficiency} 
\begin{equation}
\eta(N)=\frac{\sum_{n=1}^{N}W_n}{\max\mathcal{W}-\min\mathcal{W}}
\label{eq:spectral_filling_eff}
\end{equation}
of an energy-decoupled system. We consider only systems with equal channel widths $W_n=W,\ \forall n$. Suppose $\overline{\omega}_1=0.5W$, and let $m_n=1+(\overline{\omega}_n-0.5W)/(2W)$, so that we have
\begin{equation}
\mathcal{W}_n=\left[(2m_n-2)W, (2m_n-1)W\right],\quad m_n\ge 1.
\label{eq:constrained_channels}
\end{equation}
From~\eqref{eq:condition} and~\eqref{eq:constrained_channels}, the condition for energy-decoupling is
\begin{equation}
\left|(m_{n_1}+m_{n_2})-(m_n+m_{n_3})\right|\ge 1,\ \forall\left\{n_1, n_2\right\}\neq\left\{n, n_3\right\}
\label{eq:sidon_constraint_real}
\end{equation}
i.e., all pairwise sums of $m_n$ differ by at least $1$. This condition was derived in~\cite{forghieri_nofwm} to ensure the absence of FWM. The case where the $m_n$ are constrained to be integers was also addressed in~\cite{forghieri_nofwm}; in this case $\left(m_n\right)$ is a Sidon sequence~\cite{Sidon1932} which satisfies
\begin{equation}
m_{n_1}+m_{n_2}\neq m_n+m_{n_3},\ \forall\left\{n_1, n_2\right\}\neq\left\{n, n_3\right\}.
\label{eq:sidon_constraint}
\end{equation}
In an energy-decoupled system, the $m_n$ in~\eqref{eq:constrained_channels} must form a sequence of real numbers satisfying~\eqref{eq:sidon_constraint_real}. We call this a $\mathbb{R}$\textit{-Sidon sequence}.

The question remains how efficiently the spectrum can be filled with this method. For integer $k$, let $N(k)$ be the length of the longest $\mathbb{R}$-Sidon sequence with elements in $\left\{1,\ldots,k\right\}$:
\begin{equation}
N(k)=\max_{(m_n)} N\ \mathrm{s.t.} \begin{cases}m_n\inset{1}{k}  \forall n\inset{1}{N}  \\
 (m_n)_{n=1}^{N}\ \mathrm{is\ a\ }\mathbb{R}\mathrm{-Sidon\ sequence}. \end{cases}
\end{equation}
\begin{lemma}
	\begin{equation}
	\limsup_{k\rightarrow \infty} \frac{N(k)}{\sqrt{k}}\le 1.
	\label{eq:sidon_real_ub}
	\end{equation}
	\label{th:real_sidon}
\end{lemma}
\begin{proof}
	Erd\"os proved~\eqref{eq:sidon_real_ub} for Sidon sequences~\cite{JLMS:JLMS0212}. We extend the proof to $\mathbb{R}$-Sidon sequences. Consider the sequence $(m_n)$, with $1\le m_1<\ldots< m_N\le k$, and which satisfies~\eqref{eq:sidon_constraint_real}. Consider a positive integer $a\le k$, and define the intervals:
	\begin{equation}
	\mathcal{I}_u\triangleq\left[u-a, u\right),\quad u\in\left\{1, 2,\ldots, k+a\right\}.
	\end{equation}
	Let $M_u$ be the number of $m_i$'s in $\mathcal{I}_u$. As each $m_i$ occurs in $a$ intervals, we have
	\begin{equation}
	\sum_{u=1}^{k+a} M_u=Na.
	\label{eq:M_u}
	\end{equation}
	Let $P_u=M_u\left(M_u-1\right)/2$ be the number of pairs $(m_i, m_j),\ i<j$ in interval $\mathcal{I}_u$. Using~\eqref{eq:M_u} and
	\begin{equation}
		\sum_{u=1}^{k+a} M_u^2\ge\frac{1}{k+a}\left(\sum_{u=1}^{k+a} M_u\right)^2
	\end{equation}
	we have
	\begin{equation}
	N_p\triangleq\sum_{u=1}^{k+a} P_u \ge \frac{1}{2}Na\left(\frac{Na}{k+a}-1\right).
	\label{eq:N_p_lb}
	\end{equation}
	For $i<j$, consider the differences $d_{i, j}=m_j-m_i$. By~\eqref{eq:sidon_constraint_real}, the distance between any two $d_{i, j}$ must be least $1$. Let $r$ be an integer. If there is one $d_{i, j}$ in $\left[r, r+1\right)$, then set $d_r=d_{i, j}$; if not, set $d_r=a$. Note that $d_r\ge r$. Let $M_{i,j}$ be the number of intervals $\mathcal{I}_u$ that contain the pair $(m_i, m_j)$. We have
	\begin{align}
	N_p & =\sum_{(i,j)\colon d_{i,j}\le a}M_{i,j} \overset{(\text{a})}{\le} \sum_{r=1}^{a-1} \left\lceil a-d_r\right\rceil \nonumber \\ &  \le \sum_{r=1}^{a-1} \left(a+1-r\right)=\frac{1}{2}(a+2)(a-1)
	\label{eq:N_p_ub}
	\end{align}
	where $(\text{a})$ holds because, for $d_{i,j}\le a$, we have $M_{i,j}\le\left\lceil a-d_{i, j} \right\rceil$. Combining~\eqref{eq:N_p_lb} with~\eqref{eq:N_p_ub}, we obtain
	\begin{equation}
	N\left(\frac{Na}{k+a}-1\right)\le\frac{(a+2)(a-1)}{a}
	\end{equation}
	which implies
	\begin{equation}
	N\le \frac{1}{2}\left(1+\frac{k}{a}\right)+\sqrt{\frac{1}{4}\left(1+\frac{k}{a}\right)^2+\frac{(a+2)(a-1)(k+a)}{a^2}}.
	\end{equation}
	Choosing $a=\left\lceil k^{3/4}\right\rceil$, we have $N=\sqrt{k}+\mathcal{O}\left(k^{3/8}\right)$, which proves the lemma.	
	%	The bound $N(k)\le \tilde{N}(k)$ is straightforward because every integer-valued Sidon sequence is also a real-valued Sidon sequence. For the other bound, consider a real-valued Sidon sequence $\tilde{\mathcal{M}}^{(k)}$, with terms $\tilde{m}_n$, that achieves the maximum length $\tilde{N}(k)$. Let $m_n=\mathrm{round}(3\tilde{m}_n)$, where $\mathrm{round}(x)$ takes $x$ to the nearest integer, and let
	%	\begin{equation}
	%	\mathcal{M}^{(k)}=\left(m_n\right)_{n=1}^{\tilde{N}(k)}.
	%	\end{equation}
	%	For any $n_1, n_2, n, n_3$ we have
	%	\begin{equation}
	%	\left|(3\tilde{m}_{n_1}+3\tilde{m}_{n_2})-(3\tilde{m}_n+3\tilde{m}_{n_3})\right|\ge 3.
	%	\label{eq:sidon_constraint_3}
	%	\end{equation}	
	%	As $\left|\mathrm{round}(x)-x\right|<0.5$, rounding all terms in the left hand side of~\eqref{eq:sidon_constraint_3} can change its value by at most $2$. Therefore, we have
	%	\begin{equation}
	%	\left|(m_{n_1}+m_{n_2})-(m_n+m_{n_3})\right|\ge 1
	%	\end{equation}	
	%	so $\mathcal{M}^{(k)}$ is an integer-valued Sidon sequence with length $\tilde{N}(k)$, and whose terms are at most $3k$. Therefore, we have $\tilde{N}(k)\le N(3k)$. 
\end{proof}

%A very simple and computationally efficient method to construct a Sidon sequence was also given in~\cite{JLMS:JLMS0212}. When $N$ is a prime number, the Erd\"os sequence:
%\begin{multline}
%\mathcal{M}_{\mathrm{Erd\ddot{o}s}}(N) \\ \triangleq\left\{1+2Nn+\left(n^2 \mod N\right) \colon n\inset{0}{N-1}\right\}
%\end{multline}
%is a Sidon sequence. For example, $\mathcal{M}_{\mathrm{Erd\ddot{o}s}}(11)=\left\{1, 24, 49, 76, 94, 114, 136, 160, 186, 203, 222\right\}$. For composite $N$, one can construct $\mathcal{M}_{\mathrm{Erd\ddot{o}s}}(\overline{N})$, where $\overline{N}$ is the smallest prime larger than $N$, and then remove the $\overline{N}-N$ largest elements of the resulting sequence. Clearly, $\max\left\{\mathcal{M}_{\mathrm{Erd\ddot{o}s}}(N)\right\}\le 2N^2-N$ is true for $N$ prime and asymptotically true for all $N$, so this construction achieves: 
%\begin{equation}
%\lim_{k\rightarrow \infty} \frac{N_{\mathrm{Erd\ddot{o}s}}(k)}{\sqrt{k}}\ge \frac{1}{\sqrt{2}},
%\end{equation}
%where $N_{\mathrm{Erd\ddot{o}s}}(k)$ is the length of the largest Erd\"os sequence whose elements are at most $k$. To achieve the upper bound~\eqref{eq:sidon_ub}, 
The asymptotic bound~\eqref{eq:sidon_real_ub} is achieved by Bose's construction of Sidon sequences~\cite{bose_sidon}. When $N$ is a prime power, let $\mathbb{F}_N$ be the finite field of size $N$, and $\theta$ a generator of the extended field $\mathbb{F}_{N^2}$ (a root of an irreducible polynomial $p(x)$ of degree 2 in $\mathbb{F}_{N}$). For a set $\mathcal{S}$, let $(\mathcal{S})$ be the sequence that contains the elements of $\mathcal{S}$ in increasing order. The Bose sequence
\begin{equation}
\mathcal{M}_{\mathrm{Bose}}(N)\triangleq\left(\left\{m\inset{1}{N^2-1} \colon \theta^m-\theta\in\mathbb{F}_N \right\}\right)
\end{equation}
is a Sidon sequence. For example, for $N=11$ and $p(x)=x^2+x+7$, we choose $\theta=x$ and obtain
\begin{equation}
\mathcal{M}_{\mathrm{Bose}}(N)=\left(1, 6, 22, 62, 68, 69, 71, 88, 99, 103, 113\right).
\end{equation}
If $N$ is not a prime power, then we use the next prime power and truncate the resulting sequence. As $N(k)=\left|\mathcal{M}_{\mathrm{Bose}}(N)\right|=N$~\cite[p. 597]{sidon_constructions}, and $\max\left\{\mathcal{M}_{\mathrm{Bose}}(N)\right\}\le N^2-1$, Bose's construction achieves the upper bound~\eqref{eq:sidon_real_ub}:
\begin{equation}
\lim_{k\rightarrow \infty} \frac{N(k)}{\sqrt{k}}=\lim_{N\rightarrow \infty} \frac{N}{\sqrt{\max\mathcal{M}_\mathrm{Bose}(N)}}=1.
\label{eq:sidon_bound}
\end{equation}
\begin{remark}
	 Integer-valued Bose sequences achieve the asymptotic upper bound~\eqref{eq:sidon_real_ub}. Therefore, the $m_n$ are constrained to be integers in the following.
\end{remark}
\begin{remark}
	For small $N$, there are Sidon sequences with smaller $k$ than the Bose sequences. The optimal sequences for $N\le 16$ are listed in~\cite{shearer_golomb} and references therein. 
\end{remark}

At this point, we are ready to bound the maximum spectral filling efficiency of an energy-decoupled system.
\begin{theorem}
	Let $\mathcal{W}$ be a multi-channel frequency band, as defined in~\eqref{eq:W}, with constant channel widths $W_n=W$. Furthermore, let $\mathcal{W}$ be such that the NLSE system~\eqref{eq:nlse_f} with attenuation profile~\eqref{eq:alpha} is energy-decoupled, i.e., $\mathcal{W}$ fulfills~\eqref{eq:condition}. Then the optimal spectral filling efficiency~\eqref{eq:spectral_filling_eff} for $N$ channels belongs to $\mathcal{O}(1/N)$ and we have
	\begin{equation}
	\lim_{N\rightarrow\infty} \eta(N) N =\frac{1}{2}.
	\label{eq:efficiency}
	\end{equation}
	\label{th:efficiency}
\end{theorem}
\begin{proof}
	The spectral filling efficiency of~\eqref{eq:constrained_channels} is
	\begin{equation}
	\eta(N(k))=\frac{N(k)W}{(2k-1)W}.
	\label{eq:eta}
	\end{equation}
	Multiplying by $N(k)$ and using~\eqref{eq:sidon_bound}, we obtain~\eqref{eq:efficiency}.
\end{proof}
Theorem~\ref{th:efficiency} proves that a $N$-channel energy-decoupled system with uniform channel width $W$ can asymptotically fill at most a fraction $1/(2N)$ of the spectrum. A system with unequal channel widths might be able to achieve better spectral efficiency. For $N=2$ and $N=3$, standard optimization techniques show that equal width channels maximize the spectral efficiency. For $N>3$, the \textit{largest} channel width must be less than or equal to the \textit{smallest} empty space between channels (to see this, let $n_1=n$ be the largest channel, and $n_2$ and $n_3$ be the two closest channels in~\eqref{eq:condition}). Systems with equal channel widths have all widths equal to the minimum spacing. This does not guarantee optimal spectral efficiency, but leads us to believe that allowing unequal widths gives small improvements only.% The factor $1/2$ is due to~\eqref{eq:condition}, where the non-overlapping constraint is applied to pairwise convolutions of channels that have width $2W$.

\subsection{Numerical Results}\label{sec:numerical}

%\begin{table}[tbp]\centering
%	\caption{Fiber parameters}
%	\label{tab:parameters}
%	\begin{tabular}{|c|c|c|}
%		\hline
%		\textbf{Parameter} & \textbf{Symbol} & \textbf{Value} \\
%		Dispersion coefficient & $\beta_2$ & $\begin{array}{>{\raggedleft}p{3.7em} p{4em}}$-21.667$ & $\mathrm{ps}^2/\mathrm{km}$\end{array}$ \\
%		Nonlinear coefficient & $\gamma$ & $\begin{array}{>{\raggedleft}p{3.7em} p{4em}}$1.2578$ & $\mathrm{W}^{-1}\mathrm{km}^{-1}$\end{array}$\\
%		\hline
%	\end{tabular}
%\end{table}

Consider a $5$-channel system with  dispersion coefficient, $\beta_2=-21.667\;\mathrm{ps}^2/\mathrm{km}$, nonlinear coefficient $\gamma=1.2578\;\mathrm{W}^{-1}\mathrm{km}^{-1}$, and channel bandwidth $W/(2\pi)=1\;\mathrm{GHz}$. To demonstrate the practical use of the scheme, we drop the assumption of ideal distributed filtering and use brick-walls filters~\eqref{eq:H} every $\Deltarm z=10\;\mathrm{km}$. In a first numerical experiment, the channels are placed according to the densest Sidon sequence: $\left(1, 2, 5, 10, 12\right)$, i.e., their centers are at
\begin{equation}
\overline{\omega}_n=\left(0.5W,\ 2.5W,\ 8.5W,\   18.5W,\   22.5W\right).
\end{equation}
In a second experiment, the channels are uniformly spaced in frequency in the same bandwidth of $23W$. A scaled root raised cosine pulse with roll-off factor of $\beta=0.15$ and total bandwidth $W$ is sent in each channel. The five scaling factors are chosen randomly to have the respective phases $\{-2.397, -0.217, 2.065, 2.937, 3.003\}\;\mathrm{rad}$ and pulse energies $\{0.039, 0.468, 1.469, 1.160, 1.166\} \;\mathrm{pJ}$.

Due to the finite filter spacing, some of the launch energy is lost during propagation. In the Sidon system (Figure~\ref{fig:sidon_energy}), $2.2\%$ of the energy is lost, but the energy per channel stays approximately constant. In the uniform system (Figure~\ref{fig:nosidon_energy}), the energy loss is $0.98\%$, but the energy fluctuations between channels are apparent.

%The results show that the energy per channel stays constant in the Sidon system (Figure~\ref{fig:sidon_energy}), whereas in the uniform system (Figure~\ref{fig:nosidon_energy}), the energy fluctuations between channels are apparent.

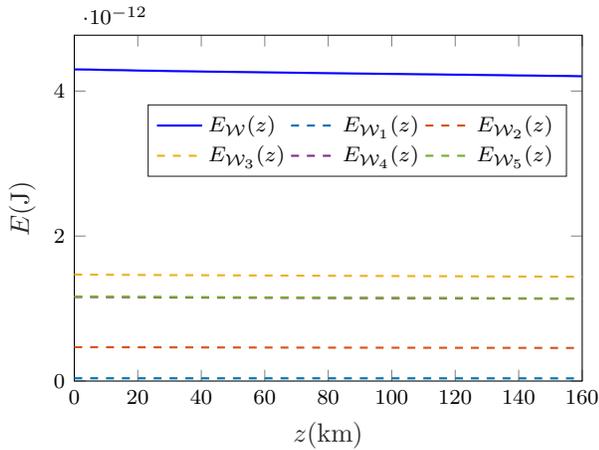
\begin{figure}[tbp]\centering
\definecolor{mycolor1}{rgb}{0.00000,0.44700,0.74100}%
\definecolor{mycolor2}{rgb}{0.85000,0.32500,0.09800}%
\definecolor{mycolor3}{rgb}{0.92900,0.69400,0.12500}%
\definecolor{mycolor4}{rgb}{0.49400,0.18400,0.55600}%
\definecolor{mycolor5}{rgb}{0.46600,0.67400,0.18800}%
\definecolor{mycolor6}{rgb}{0.30100,0.74500,0.93300}%
\begin{tikzpicture}

\begin{axis}[%
width=0.951\figurewidth,
height=\figureheight,
at={(0\figurewidth,0\figureheight)},
scale only axis,
xmin=0,
xmax=160,
xlabel style={font=\color{white!15!black}},
xlabel={$z (\mathrm{km})$},
ymin=0,
ymax=4.77e-12,
ylabel style={font=\color{white!15!black}},
ylabel={$E (\mathrm{J})$},
axis background/.style={fill=white},
legend columns=2,
transpose legend,
legend style={at={(0.97, 0.80)}, anchor=north east, legend cell align=left, align=left, draw=white!15!black}
]
\addplot [thick, color=blue]
table[row sep=crcr]{%
	0	4.29872892513732e-12\\
	5.16129032258065	4.29638223588023e-12\\
	10.3225806451613	4.2902845405314e-12\\
	15.4838709677419	4.28771523684408e-12\\
	20.6451612903226	4.28221732512394e-12\\
	25.8064516129032	4.2796238843372e-12\\
	30.9677419354839	4.27496806636521e-12\\
	36.1290322580645	4.27238634491443e-12\\
	41.2903225806452	4.26858824522703e-12\\
	46.4516129032258	4.26619269841968e-12\\
	51.6129032258065	4.26266896925682e-12\\
	56.7741935483871	4.25996616396034e-12\\
	61.9354838709677	4.25682478706673e-12\\
	67.0967741935484	4.25411403253117e-12\\
	72.258064516129	4.25100612307504e-12\\
	77.4193548387097	4.24801502061548e-12\\
	82.5806451612903	4.24537549012034e-12\\
	87.741935483871	4.24253938868002e-12\\
	92.9032258064516	4.24002949387984e-12\\
	98.0645161290323	4.23699345377962e-12\\
	103.225806451613	4.23490200061056e-12\\
	108.387096774194	4.23194379382594e-12\\
	113.548387096774	4.22984097375917e-12\\
	118.709677419355	4.22650875772486e-12\\
	123.870967741935	4.22468345117987e-12\\
	129.032258064516	4.22125716160391e-12\\
	134.193548387097	4.21930556360786e-12\\
	139.354838709677	4.21536080419251e-12\\
	144.516129032258	4.21371525129417e-12\\
	149.677419354839	4.20978401839027e-12\\
	154.838709677419	4.20810277939997e-12\\
	160	4.2041342694801e-12\\
};
\addlegendentry{$E_{\mathcal{W}}(z)$}

\addplot [thick, color=mycolor3, dashed]
table[row sep=crcr]{%
	0	1.46829419084934e-12\\
	5.16129032258065	1.46754488145037e-12\\
	10.3225806451613	1.46567916823899e-12\\
	15.4838709677419	1.46488920326912e-12\\
	20.6451612903226	1.46327671181621e-12\\
	25.8064516129032	1.46250009039082e-12\\
	30.9677419354839	1.46115549992109e-12\\
	36.1290322580645	1.46035931354252e-12\\
	41.2903225806452	1.45920339776407e-12\\
	46.4516129032258	1.45841296322329e-12\\
	51.6129032258065	1.45726152443274e-12\\
	56.7741935483871	1.45634552141373e-12\\
	61.9354838709677	1.45530661194919e-12\\
	67.0967741935484	1.45441685249578e-12\\
	72.258064516129	1.45342541321546e-12\\
	77.4193548387097	1.45248090157402e-12\\
	82.5806451612903	1.45165348562808e-12\\
	87.741935483871	1.4507416299044e-12\\
	92.9032258064516	1.44991989667986e-12\\
	98.0645161290323	1.44888242673803e-12\\
	103.225806451613	1.44815649472412e-12\\
	108.387096774194	1.44711726926108e-12\\
	113.548387096774	1.44638613244577e-12\\
	118.709677419355	1.44525650971618e-12\\
	123.870967741935	1.44465565003836e-12\\
	129.032258064516	1.44354999304971e-12\\
	134.193548387097	1.4429295630262e-12\\
	139.354838709677	1.44165468639845e-12\\
	144.516129032258	1.44112125306043e-12\\
	149.677419354839	1.43979371151865e-12\\
	154.838709677419	1.43922121877223e-12\\
	160	1.43782375643155e-12\\
};
\addlegendentry{$E_{\mathcal{W}_3}(z)$}

\addplot [thick, color=mycolor1, dashed]
table[row sep=crcr]{%
	0	3.8699969467276e-14\\
	5.16129032258065	3.86508489758845e-14\\
	10.3225806451613	3.8490606514462e-14\\
	15.4838709677419	3.84207409877699e-14\\
	20.6451612903226	3.82603173782708e-14\\
	25.8064516129032	3.82040865396048e-14\\
	30.9677419354839	3.81067566219707e-14\\
	36.1290322580645	3.8087893679261e-14\\
	41.2903225806452	3.80725373923485e-14\\
	46.4516129032258	3.80803906945743e-14\\
	51.6129032258065	3.80988303524105e-14\\
	56.7741935483871	3.81021805034681e-14\\
	61.9354838709677	3.809898995631e-14\\
	67.0967741935484	3.80780685646574e-14\\
	72.258064516129	3.80383649879599e-14\\
	77.4193548387097	3.79954936154405e-14\\
	82.5806451612903	3.79434894699861e-14\\
	87.741935483871	3.79017942142493e-14\\
	92.9032258064516	3.78566157261705e-14\\
	98.0645161290323	3.78220731359411e-14\\
	103.225806451613	3.77954741299531e-14\\
	108.387096774194	3.77688220300128e-14\\
	113.548387096774	3.77458533352336e-14\\
	118.709677419355	3.77063842383098e-14\\
	123.870967741935	3.76761750510455e-14\\
	129.032258064516	3.76096069289888e-14\\
	134.193548387097	3.75615956652125e-14\\
	139.354838709677	3.74569903491388e-14\\
	144.516129032258	3.74094365464439e-14\\
	149.677419354839	3.73054972040218e-14\\
	154.838709677419	3.72643476552422e-14\\
	160	3.71943552649026e-14\\
};
\addlegendentry{$E_{\mathcal{W}_1}(z)$}

\addplot [thick, color=mycolor4, dashed]
table[row sep=crcr]{%
	0	1.15892529259477e-12\\
	5.16129032258065	1.15830498286987e-12\\
	10.3225806451613	1.15665124779716e-12\\
	15.4838709677419	1.15596232270299e-12\\
	20.6451612903226	1.15445259701323e-12\\
	25.8064516129032	1.15374181881602e-12\\
	30.9677419354839	1.15245571908465e-12\\
	36.1290322580645	1.15175627894866e-12\\
	41.2903225806452	1.15073503332714e-12\\
	46.4516129032258	1.15011480882206e-12\\
	51.6129032258065	1.14920449995829e-12\\
	56.7741935483871	1.14851670606622e-12\\
	61.9354838709677	1.14770421611408e-12\\
	67.0967741935484	1.14698976776717e-12\\
	72.258064516129	1.14615827476069e-12\\
	77.4193548387097	1.14534494241748e-12\\
	82.5806451612903	1.14463565358631e-12\\
	87.741935483871	1.14388810036879e-12\\
	92.9032258064516	1.14325264876815e-12\\
	98.0645161290323	1.14251667552399e-12\\
	103.225806451613	1.14202761939264e-12\\
	108.387096774194	1.14133856274332e-12\\
	113.548387096774	1.14084781634398e-12\\
	118.709677419355	1.1400360432728e-12\\
	123.870967741935	1.1395812374457e-12\\
	129.032258064516	1.13870622223212e-12\\
	134.193548387097	1.13820691585465e-12\\
	139.354838709677	1.13722654166993e-12\\
	144.516129032258	1.136821914395e-12\\
	149.677419354839	1.13591515470644e-12\\
	154.838709677419	1.13552410435248e-12\\
	160	1.13464868658854e-12\\
};
\addlegendentry{$E_{\mathcal{W}_4}(z)$}

\addplot [thick, color=mycolor2, dashed]
table[row sep=crcr]{%
	0	4.6772058856366e-13\\
	5.16129032258065	4.67387228874959e-13\\
	10.3225806451613	4.66482066864132e-13\\
	15.4838709677419	4.66135496524853e-13\\
	20.6451612903226	4.65347619410085e-13\\
	25.8064516129032	4.65008995130274e-13\\
	30.9677419354839	4.64348728640138e-13\\
	36.1290322580645	4.64006985173897e-13\\
	41.2903225806452	4.6345676411355e-13\\
	46.4516129032258	4.63135207644424e-13\\
	51.6129032258065	4.6262686796878e-13\\
	56.7741935483871	4.62279780580008e-13\\
	61.9354838709677	4.61859429265409e-13\\
	67.0967741935484	4.61550419949978e-13\\
	72.258064516129	4.61185706400479e-13\\
	77.4193548387097	4.60881198657156e-13\\
	82.5806451612903	4.60602551447746e-13\\
	87.741935483871	4.60318129172024e-13\\
	92.9032258064516	4.60051699437611e-13\\
	98.0645161290323	4.59716644880777e-13\\
	103.225806451613	4.59474839096216e-13\\
	108.387096774194	4.5912188031845e-13\\
	113.548387096774	4.58868859033828e-13\\
	118.709677419355	4.58476314516636e-13\\
	123.870967741935	4.58267022374152e-13\\
	129.032258064516	4.57889135980323e-13\\
	134.193548387097	4.57677144277863e-13\\
	139.354838709677	4.57250104599215e-13\\
	144.516129032258	4.57065936100852e-13\\
	149.677419354839	4.56612800638124e-13\\
	154.838709677419	4.56409349921584e-13\\
	160	4.55915559120194e-13\\
};
\addlegendentry{$E_{\mathcal{W}_2}(z)$}

\addplot [thick, color=mycolor5, dashed]
table[row sep=crcr]{%
	0	1.16508888366228e-12\\
	5.16129032258065	1.16449429370915e-12\\
	10.3225806451613	1.16298145111665e-12\\
	15.4838709677419	1.16230747335935e-12\\
	20.6451612903226	1.16088007950615e-12\\
	25.8064516129032	1.16016889346049e-12\\
	30.9677419354839	1.15890136209736e-12\\
	36.1290322580645	1.15817587357009e-12\\
	41.2903225806452	1.15712051262993e-12\\
	46.4516129032258	1.15644932803534e-12\\
	51.6129032258065	1.15547724654459e-12\\
	56.7741935483871	1.15472197539692e-12\\
	61.9354838709677	1.15385553978174e-12\\
	67.0967741935484	1.15307892375358e-12\\
	72.258064516129	1.15219836371045e-12\\
	77.4193548387097	1.15131248435139e-12\\
	82.5806451612903	1.15054030998822e-12\\
	87.741935483871	1.14968973502056e-12\\
	92.9032258064516	1.14894863326805e-12\\
	98.0645161290323	1.14805563350088e-12\\
	103.225806451613	1.14744757326763e-12\\
	108.387096774194	1.14659725947307e-12\\
	113.548387096774	1.14599231260037e-12\\
	118.709677419355	1.14503350598094e-12\\
	123.870967741935	1.14450336627061e-12\\
	129.032258064516	1.14350220341276e-12\\
	134.193548387097	1.14293034478393e-12\\
	139.354838709677	1.14177248117578e-12\\
	144.516129032258	1.14129671119145e-12\\
	149.677419354839	1.14015685432303e-12\\
	154.838709677419	1.13968375869843e-12\\
	160	1.13855191207492e-12\\
};
\addlegendentry{$E_{\mathcal{W}_5}(z)$}

\end{axis}
\end{tikzpicture}%
	\caption{Energy evolution in a $5$-channel system using a Sidon sequence and filters every $10\;\mathrm{km}$.}
	\label{fig:sidon_energy}
\end{figure}

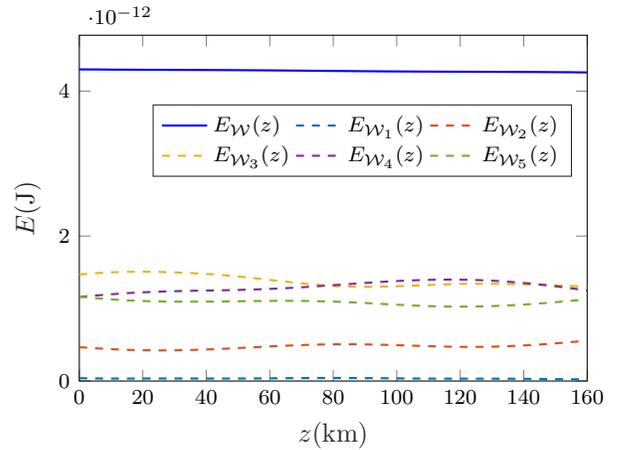
\begin{figure}[tbp]\centering
\definecolor{mycolor1}{rgb}{0.00000,0.44700,0.74100}%
\definecolor{mycolor2}{rgb}{0.85000,0.32500,0.09800}%
\definecolor{mycolor3}{rgb}{0.92900,0.69400,0.12500}%
\definecolor{mycolor4}{rgb}{0.49400,0.18400,0.55600}%
\definecolor{mycolor5}{rgb}{0.46600,0.67400,0.18800}%
\definecolor{mycolor6}{rgb}{0.30100,0.74500,0.93300}%
\begin{tikzpicture}

\begin{axis}[%
width=0.951\figurewidth,
height=\figureheight,
at={(0\figurewidth,0\figureheight)},
scale only axis,
xmin=0,
xmax=160,
xlabel style={font=\color{white!15!black}},
xlabel={$z (\mathrm{km})$},
ymin=0,
ymax=4.77e-12,
ylabel style={font=\color{white!15!black}},
ylabel={$E (\mathrm{J})$},
axis background/.style={fill=white},
legend columns=2,
transpose legend,
legend style={at={(0.97, 0.80)}, anchor=north east, legend cell align=left, align=left, draw=white!15!black}
]
\addplot [thick, color=blue]
table[row sep=crcr]{%
	0	4.29876845362861e-12\\
	5.16129032258065	4.29777725191981e-12\\
	10.3225806451613	4.29558414387466e-12\\
	15.4838709677419	4.29489575057348e-12\\
	20.6451612903226	4.29346841981555e-12\\
	25.8064516129032	4.29283736321334e-12\\
	30.9677419354839	4.29165678368212e-12\\
	36.1290322580645	4.29095117377231e-12\\
	41.2903225806452	4.28982203777727e-12\\
	46.4516129032258	4.28896337018894e-12\\
	51.6129032258065	4.28755203196987e-12\\
	56.7741935483871	4.28619600793063e-12\\
	61.9354838709677	4.28447082062525e-12\\
	67.0967741935484	4.28272956238454e-12\\
	72.258064516129	4.28062582751251e-12\\
	77.4193548387097	4.27846136631265e-12\\
	82.5806451612903	4.27653175622203e-12\\
	87.741935483871	4.27452531258273e-12\\
	92.9032258064516	4.27281108787336e-12\\
	98.0645161290323	4.27099374700694e-12\\
	103.225806451613	4.26984795334945e-12\\
	108.387096774194	4.26852576749414e-12\\
	113.548387096774	4.26769656714524e-12\\
	118.709677419355	4.26663738619321e-12\\
	123.870967741935	4.2661077561996e-12\\
	129.032258064516	4.26515604559497e-12\\
	134.193548387097	4.26457495189522e-12\\
	139.354838709677	4.2632041215199e-12\\
	144.516129032258	4.26251837802123e-12\\
	149.677419354839	4.26055592805749e-12\\
	154.838709677419	4.2595410628064e-12\\
	160	4.25665983072256e-12\\
};
\addlegendentry{$E_{\mathcal{W}}(z)$}

\addplot [thick, color=mycolor3, dashed]
table[row sep=crcr]{%
	0	1.46928365584481e-12\\
	5.16129032258065	1.4876624710553e-12\\
	10.3225806451613	1.49990645194029e-12\\
	15.4838709677419	1.50621542128315e-12\\
	20.6451612903226	1.50731164522674e-12\\
	25.8064516129032	1.50426652668886e-12\\
	30.9677419354839	1.49763432690021e-12\\
	36.1290322580645	1.4871636711991e-12\\
	41.2903225806452	1.47374183931884e-12\\
	46.4516129032258	1.4564867759117e-12\\
	51.6129032258065	1.43504560711695e-12\\
	56.7741935483871	1.41027507011131e-12\\
	61.9354838709677	1.38529692360742e-12\\
	67.0967741935484	1.36029002016406e-12\\
	72.258064516129	1.3369436824427e-12\\
	77.4193548387097	1.31792224657008e-12\\
	82.5806451612903	1.30561330738216e-12\\
	87.741935483871	1.29941670126724e-12\\
	92.9032258064516	1.29931173208977e-12\\
	98.0645161290323	1.30451909361871e-12\\
	103.225806451613	1.31258397699154e-12\\
	108.387096774194	1.32198145251368e-12\\
	113.548387096774	1.33095823550167e-12\\
	118.709677419355	1.33779341357908e-12\\
	123.870967741935	1.34136201084911e-12\\
	129.032258064516	1.34177810034646e-12\\
	134.193548387097	1.3393687470384e-12\\
	139.354838709677	1.33425686138729e-12\\
	144.516129032258	1.32801307533697e-12\\
	149.677419354839	1.32014221927542e-12\\
	154.838709677419	1.3113202641599e-12\\
	160	1.30131155879097e-12\\
};
\addlegendentry{$E_{\mathcal{W}_3}(z)$}

\addplot [thick, color=mycolor1, dashed]
table[row sep=crcr]{%
	0	3.85315716385207e-14\\
	5.16129032258065	3.55541928387826e-14\\
	10.3225806451613	3.41594878057226e-14\\
	15.4838709677419	3.40283060771969e-14\\
	20.6451612903226	3.44606910536412e-14\\
	25.8064516129032	3.48585132920511e-14\\
	30.9677419354839	3.49442060973627e-14\\
	36.1290322580645	3.47066765377676e-14\\
	41.2903225806452	3.44157855785084e-14\\
	46.4516129032258	3.43945403640241e-14\\
	51.6129032258065	3.49198079305584e-14\\
	56.7741935483871	3.61139383442779e-14\\
	61.9354838709677	3.77185003408576e-14\\
	67.0967741935484	3.94897760862501e-14\\
	72.258064516129	4.10493374778541e-14\\
	77.4193548387097	4.1943620427368e-14\\
	82.5806451612903	4.19276353455098e-14\\
	87.741935483871	4.10239249489955e-14\\
	92.9032258064516	3.93552691359414e-14\\
	98.0645161290323	3.7261490688093e-14\\
	103.225806451613	3.53160914200333e-14\\
	108.387096774194	3.37265813511286e-14\\
	113.548387096774	3.27130175360786e-14\\
	118.709677419355	3.23361172649675e-14\\
	123.870967741935	3.23793567086096e-14\\
	129.032258064516	3.24476591254976e-14\\
	134.193548387097	3.21376850079807e-14\\
	139.354838709677	3.10458578619272e-14\\
	144.516129032258	2.92444306461036e-14\\
	149.677419354839	2.6808357626048e-14\\
	154.838709677419	2.42944314304992e-14\\
	160	2.25218452116527e-14\\
};
\addlegendentry{$E_{\mathcal{W}_1}(z)$}

\addplot [thick, color=mycolor4, dashed]
table[row sep=crcr]{%
	0	1.15982408514161e-12\\
	5.16129032258065	1.17981364404687e-12\\
	10.3225806451613	1.19774004307605e-12\\
	15.4838709677419	1.21328907963327e-12\\
	20.6451612903226	1.22503756769757e-12\\
	25.8064516129032	1.23416710794649e-12\\
	30.9677419354839	1.24037426908472e-12\\
	36.1290322580645	1.24527605394944e-12\\
	41.2903225806452	1.24898211783493e-12\\
	46.4516129032258	1.25315257341592e-12\\
	51.6129032258065	1.25837554414135e-12\\
	56.7741935483871	1.26564181293887e-12\\
	61.9354838709677	1.27451427208836e-12\\
	67.0967741935484	1.28573802324275e-12\\
	72.258064516129	1.29931373269139e-12\\
	77.4193548387097	1.3147083142919e-12\\
	82.5806451612903	1.33016524023664e-12\\
	87.741935483871	1.34586182895592e-12\\
	92.9032258064516	1.36128775451197e-12\\
	98.0645161290323	1.37511640812563e-12\\
	103.225806451613	1.38598586885235e-12\\
	108.387096774194	1.39392034710645e-12\\
	113.548387096774	1.39855405950006e-12\\
	118.709677419355	1.39887013533995e-12\\
	123.870967741935	1.39511378691496e-12\\
	129.032258064516	1.38674435555898e-12\\
	134.193548387097	1.37371298550203e-12\\
	139.354838709677	1.35546816121266e-12\\
	144.516129032258	1.33451114293903e-12\\
	149.677419354839	1.30906912609792e-12\\
	154.838709677419	1.2803820228127e-12\\
	160	1.24937013547005e-12\\
};
\addlegendentry{$E_{\mathcal{W}_4}(z)$}

\addplot [thick, color=mycolor2, dashed]
table[row sep=crcr]{%
	0	4.67021610087987e-13\\
	5.16129032258065	4.51770326104894e-13\\
	10.3225806451613	4.38876497073613e-13\\
	15.4838709677419	4.29465633309761e-13\\
	20.6451612903226	4.23855613871316e-13\\
	25.8064516129032	4.22110176095543e-13\\
	30.9677419354839	4.23963185860633e-13\\
	36.1290322580645	4.29352046627777e-13\\
	41.2903225806452	4.37015054051403e-13\\
	46.4516129032258	4.46815916022253e-13\\
	51.6129032258065	4.58126791332812e-13\\
	56.7741935483871	4.70132628851406e-13\\
	61.9354838709677	4.81061170279224e-13\\
	67.0967741935484	4.91010021073166e-13\\
	72.258064516129	4.99198045718081e-13\\
	77.4193548387097	5.04752343322575e-13\\
	82.5806451612903	5.07089639567228e-13\\
	87.741935483871	5.06426937234199e-13\\
	92.9032258064516	5.02918146566024e-13\\
	98.0645161290323	4.96950937809095e-13\\
	103.225806451613	4.90120966536818e-13\\
	108.387096774194	4.82859728186277e-13\\
	113.548387096774	4.76381520971989e-13\\
	118.709677419355	4.71820964266942e-13\\
	123.870967741935	4.70424019110333e-13\\
	129.032258064516	4.72443419075214e-13\\
	134.193548387097	4.78568394345799e-13\\
	139.354838709677	4.88992392867554e-13\\
	144.516129032258	5.02316475973896e-13\\
	149.677419354839	5.18705503465683e-13\\
	154.838709677419	5.37234633592622e-13\\
	160	5.55899496776626e-13\\
};
\addlegendentry{$E_{\mathcal{W}_2}(z)$}

\addplot [thick, color=mycolor5, dashed]
table[row sep=crcr]{%
	0	1.16410753091568e-12\\
	5.16129032258065	1.14297661787395e-12\\
	10.3225806451613	1.12490166397899e-12\\
	15.4838709677419	1.1118973102701e-12\\
	20.6451612903226	1.10280290196628e-12\\
	25.8064516129032	1.09743503919039e-12\\
	30.9677419354839	1.0947407957392e-12\\
	36.1290322580645	1.09445272545822e-12\\
	41.2903225806452	1.09566724099359e-12\\
	46.4516129032258	1.09811356447505e-12\\
	51.6129032258065	1.10108428144821e-12\\
	56.7741935483871	1.10403255768477e-12\\
	61.9354838709677	1.1058799543094e-12\\
	67.0967741935484	1.10620172181832e-12\\
	72.258064516129	1.10412102918249e-12\\
	77.4193548387097	1.09913484170073e-12\\
	82.5806451612903	1.09173593369049e-12\\
	87.741935483871	1.08179592017637e-12\\
	92.9032258064516	1.06993818556966e-12\\
	98.0645161290323	1.05714581676541e-12\\
	103.225806451613	1.04584104954871e-12\\
	108.387096774194	1.03603765833661e-12\\
	113.548387096774	1.02908973363544e-12\\
	118.709677419355	1.02581675574227e-12\\
	123.870967741935	1.02682858261658e-12\\
	129.032258064516	1.03174251148881e-12\\
	134.193548387097	1.04078714000101e-12\\
	139.354838709677	1.05344084819047e-12\\
	144.516129032258	1.06843325312524e-12\\
	149.677419354839	1.08583072159242e-12\\
	154.838709677419	1.10630971081068e-12\\
	160	1.12755679447326e-12\\
};
\addlegendentry{$E_{\mathcal{W}_5}(z)$}

\end{axis}
\end{tikzpicture}%
	\caption{Energy evolution in a $5$-channel system with uniform channel spacing in frequency and filters every $10\;\mathrm{km}$.}
	\label{fig:nosidon_energy}
\end{figure}

\section{Conclusion}\label{sec:conclusion}
 
We have proposed a distributed filtering approach to mitigate spectral broadening in the NLSE. We have proved that the new model preserves energy, thus establishing an invariant that can be useful for communications.

We have characterized the evolution of the per-channel energy in a WDM system with distributed filtering. We have derived conditions that ensure per-channel energy conservation by using Sidon sequences. For constant channel widths, we have proved that an $N$-channel Sidon system can asymptotically fill at most a fraction $1/(2N)$ of the spectrum, which implies that spectral efficiency goes to $0$ as $N$ increases. The capacity analysis of this new band-limited channel model is an interesting open problem.

\section{Acknowledgment}
The authors would like to thank R.-J. Essiambre, A. Mecozzi, and M. Shtaif for stimulating discussions on energy conservation in fiber with bandwidth constraints.


%\section*{Acknowledgment}



\bibliographystyle{IEEEtran}
\bibliography{energy_conservation_r1}
%\begin{thebibliography}{1}
  %0.5em minus 0.4em\relax Harlow, England: Addison-Wesley, 1999.
%
%\end{thebibliography}

\end{document}